%% file: anonymous-submission-latex-2026.tex
\documentclass[accepted]{uai2026} 

\usepackage[american]{babel}

\usepackage{natbib}
\usepackage{mathtools} 
\usepackage{amssymb}
\usepackage{amsthm}
\usepackage{booktabs}
\usepackage{graphicx}
\usepackage{subcaption}
\usepackage{algorithm}
\usepackage{algpseudocode}

\theoremstyle{plain}
\newtheorem{theorem}{Theorem}[section]

\theoremstyle{definition}

\theoremstyle{remark}

\title{Why Weighting Matters: A Theoretical Perspective on Link Prediction Heuristics}

\author[1]{Jiajie Shi}
\author[1]{David A. Meyer}
\affil[1]{Department of Mathematics\\University of California San Diego}

\begin{document}
\maketitle
\pagestyle{numbered}

\begin{abstract}
In network analysis, nodes with different degrees are usually weighted differently in tasks such as link prediction. For example, two individuals connected to the same high-degree node (e.g., a celebrity) are often considered less likely to know each other than if they share a low-degree connection. In this paper, we theoretically justify the reasoning behind this heuristic using a graph generation model. In this model, nodes are embedded in a latent metric space, and edges are more likely to form between nearby nodes. We analytically derive degree-based weighting schemes for commonly used link prediction and similarity measures. In particular, we identify the optimal weighting scheme for the Common Neighbors index and extend our analysis to other popular heuristics, including the graph distance. Our theoretical findings align with previous empirical observations, and we further validate them through experiments on both synthetic and real-world networks.
\end{abstract}

\section{Introduction}

Link prediction is an important problem in the fields of network science and graph analysis. Given a graph without extra node features, the key for prediction is to assess the similarity of pairs, which quantifies the resemblance between two nodes. Many heuristic measures have been suggested to measure similarity, such as common neighbors \citep{newman2001clustering}, preferential attachment \citep{barabasi2002evolution}, and the Katz index \citep{katz1953new}.

The number of common neighbors of two nodes, defined as the number of nodes connected to both, has frequently served as a historical benchmark for link prediction, as well as a similarity measure. Empirical findings indicate that assigning lower weights to common nodes with higher degrees brings better link prediction results \citep{adamic2003friends,zhou2009predicting}. Recent studies also applied different weight schemes to other similarity indices such as the number of \(l\)-hop paths \citep{kovacs2019network} and the Katz index \citep{aziz2020path}.
This raises fundamental questions:

\emph{Why degree weighting matters in link prediction? And how should we weighting them?}

While many link prediction heuristics are well-supported by empirical evidence, there is far less theoretical understanding of their effectiveness. In particular, the role of degree-based weighting remains heuristic in nature. This paper provides a principled justification for such techniques through analysis in a latent space model.

Geometry has demonstrated the potential to help understand complex networks. Latent space models, particularly in Euclidean space, have been used to model network structures, as Sarkar et al. \citep{sarkar2011theoretical}, for example, derived error bounds for some popular link prediction indices within the Euclidean random graph model. However, the Euclidean graph models fail to capture key features of real networks, such as the small world phenomenon, scale-freeness, and power law degree distributions, and also does not explain why some of the link prediction methods outperform others. Krioukov et al.~\citep{krioukov2010hyperbolic} introduced a hyperbolic random graph (HRG) that can replicate the majority of the aforementioned properties, utilizing a relatively straightforward approach in hyperbolic space. Further studies on hyperbolic random graphs have also found properties that largely align with empirical findings on real-world networks, such as the existence of a giant component ~\citep{bode2013giant}, diameter ~\citep{kiwi2014bound, friedrich2018diameter}, average distance ~\citep{abdullah1506typical}, separators and treewidth ~\citep{blasius2016hyperbolic}, clique number ~\citep{friedrich2015cliques}, etc. These findings motivate the use of hyperbolic models for analyzing complex networks, as they reflect the structural patterns commonly observed in practice.
\subsection{Main Contribution and Organization}
In this paper, we address the two fundamental questions posed earlier by analyzing similarity indices in hyperbolic space. Our primary contributions are as follows:
\begin{itemize}
  \item 
  We show how the number of common neighbors relates to latent distance in hyperbolic space. We demonstrate that, by carefully weighting common neighbors based on their degrees, they can be more effectively used to estimate latent similarities.
  \item We extend our analysis to include some global indices such as the number of \(l\)-hop neighbors, yielding a weighted version of the \(l\)-hop indices and the Katz index.
  \item A modified similarity index is proposed based on the shortest path: HGD (hyperbolic graph distance), which provides more accurate estimation of long-range similarities than standard graph distance.
\end{itemize}
\textbf{Organization of this paper}:\\
As a preliminary, we introduce the hyperbolic graph model and basic properties in Section 3. Section 4 analyzes local similarity indices based on common neighbors and derives the optimal weighting. Section 5 extends the degree-weighting heuristic to global indices—including graph distance and the number of $l$-hop paths—and shows that these indices also benefit from appropriate weighting. Section 6 presents empirical validation and discusses potential applications.

\section{Related Work}
The concept of using similarity to estimate the likelihood of a link can be traced back to McPherson et al.~\citep{mcpherson2001birds}, who proposed ``Similarity breeds connection''. Over the years, many indices have been developed to measure the similarity between pairs of nodes in a network and thence used for link prediction. These include the Shortest Path~\citep{liben2003link}, which also can be called graph distance; the Number of Common Neighbors, first proposed by Newman~\citep{newman2001clustering} when studying the edges of co-authorship networks; the Adamic \& Adar index~\citep{adamic2003friends}, a weighted version of the common neighbor index; and the Preferential Attachment index~\citep{barabasi2002evolution}, from a model generating growing scale-free networks. More recently, variants of these indices keep emerging to improve link prediction performance: Kovács et al.~\citep{kovacs2019network} weighted 3-hop paths by the square root of the product of the intermediate nodes' degrees when predicting protein-protein interactions; and Aziz et al.~\citep{aziz2020path} proposed global and quasi-local extensions of some commonly used local similarity indices. And besides similarity based methods, there are other approaches to link prediction, such as random walk-based methods~\citep{tong2006fast}, embedding-based methods~\citep{belkin2001laplacian,tenenbaum2000global,kitsak2020link}, etc.

In the analysis of networks, homophily is one of the most important assumptions. Hoff et al. \citep{hoff2002latent} proposed a model where nodes are distributed in a $d$-dimensional Euclidean space, where the probability of forming a link is high if two nodes are close in latent space. Under this setup, we are able to quantify similarities within the network through distances between nodes in a low-dimensional Euclidean space. From a practical perspective, the coordinates can be seen as a embedding of nodes within a latent feature space.

Besides homophily, popularity has also been considered a factor contributed to connection. For example, the preferential attachment principle \citep{barabasi2002evolution}, known as ``rich get richer'', was proposed to explain the scale-free property observed in many networks. While the latent space models are limited to Euclidean space in past decades,  Krioukov et al.~\citep{krioukov2010hyperbolic} took a different approach by assuming a hyperbolic geometry underlying the network. Papadopoulos et al.~\citep{papadopoulos2012popularity} highlights that the hyperbolic model reveals additional factors influencing link formation in social networks beyond mere similarity, specifically popularity. The hyperbolic distance between nodes can be seen as a unified metric that encapsulates both popularity and similarity. The networks generated in this setting turned out to reproduce not only inhomogeneous degree distributions but also the strong clustering and small-worldness observed in many real-world networks. Some properties of hyperbolic random graphs have been derived rigorously~\citep{gugelmann2012random,friedrich2018diameter}. 

To our knowledge, this is the first theoretical analysis of popular link prediction heuristics using hyperbolic random graphs. Our results provide a unified explanation for several empirical phenomena commonly observed in real-world networks studies.

\section{Preliminaries}
\subsection{Hyperbolic Geometry}
Hyperbolic geometry is a negatively curved geometry that differs from Euclidean geometry. There are several ways to model a hyperbolic space, including the hyperboloid model, the Beltrami-Klein model, and the Poincaré ball model. Among these models, we will use the Poincaré ball, as it imposes a relatively simple constraint on the representations and is convenient to visualize.

To characterize a two-dimensional Poincaré disk, we adopt the notation introduced in \citep{krioukov2010hyperbolic}. In this notation, the polar coordinates \((r_i, \theta_i)\) specify the position of a node \(i\), with radial coordinate \(r_i \in [0,\infty)\) representing the hyperbolic distance of \(i\) from the origin. Assuming the curvature \(K=-\zeta^2\) where \(\zeta>0\), the hyperbolic distance \(d_{ij}\) between two points with coordinates \((r_i,\theta_i)\) and \((r_j,\theta_j)\) is given by the hyperbolic law of cosines:
{\fontsize{8.5}{8.5}\selectfont
\begin{equation*}
\begin{split}
\cosh (\zeta d_{ij})=\cosh (\zeta r_i) \cosh (\zeta r_j) - \sinh (\zeta r_i )\sinh (\zeta r_j ) \cos \Delta \theta_{i,j},
\end{split}
\end{equation*}
}

where $\Delta\theta_{i,j}= \pi-|\pi-|\theta_i-\theta_j||$. There is an approximation that simplifies calculations involving the hyperbolic distance:
\begin{equation}
d_{ij}\approx r_i+r_j+\frac{2}{\zeta}\log \frac{\Delta \theta_{i,j}}{2}.
\end{equation}
The approximation is close for sufficient large $\zeta r_i$, $\zeta r_j$, and $\Delta \theta_{i,j}> 2\sqrt{e^{-2\zeta r_i}+e^{-2\zeta r_j}}$. Theoretical derivation and error bounds for this approximation are provided in \citep{gugelmann2012random}.
Empirical results from \citet{krioukov2010hyperbolic} demonstrate that it is reasonably accurate in practice, and it has been widely adopted in studies involving hyperbolic distance \citep{blasius2018efficient_sp, krioukov2010hyperbolic,blasius2018efficient}. This formulation makes it evident that three key factors influence the distance: the centrality of node $i$, centrality of node $j$, and the angular similarity between them.
\subsection{Hyperbolic Random Graphs}
The hyperbolic random graphs described in \citep{krioukov2010hyperbolic} is a model that consists of a quasi-uniform distribution of $N$ points into a 2-dimensional Poincaré disk of radius $R$ around the origin, where the curvature is fixed to be $-1$. Two points are connected if and only if their hyperbolic distance is no greater than \(R\). More precisely, the angular distribution is uniform: $\rho(\theta)=1/(2\pi)$, and the radial probability density function is:
\[
\rho(r)=\alpha \frac{\sinh \alpha r}{\cosh \alpha R-1}\approx \alpha e^{\alpha(r-R)},
\]
where \(\alpha>\frac{1}{2}\). When \(\alpha=1\), the node density is exactly uniform.
The connection function is:
\(
p(d_{ij}) = \begin{cases}
  1  & \textrm{if } d_{ij}\leq R \\
  0 & \textrm{if } d_{ij}> R.
\end{cases}
\)
The model generates a graph with a power-law degree distribution, demonstrating that network heterogeneity emerges naturally in latent hyperbolic space.

In this model, nodes closer to the origin are more likely to connect with others. A natural interpretation is that node popularity is represented by radial position, while similarity is captured by angular proximity. Hyperbolic distance thus combines these two attributes—radial popularity and angular similarity (homophily)—into a single unified metric \citep{papadopoulos2012popularity}. This formulation allows the model to replicate many key properties of real-world networks, making it an ideal framework for studying their structure and mechanisms. Within this setting, we analyze how link prediction indices reflect latent distances, thereby influencing the probability of edge formation.

Additionally, a more general model was proposed in the paper \citep{krioukov2010hyperbolic}, using a probabilistic connection. In this paper we will focus on the simplest model where the connection function is deterministic.
\subsection{Parameter Estimation}
Given a hyperbolic random graph, we only have information about connections. In order to study the latent geometry, we can start with estimating the parameters controlling the graph. Good analytical estimates have been derived in the previous work~\citep{boguna2010sustaining}.
\\
\textbf{Estimating \(\alpha\) and \(R\)}\\
The parameter \(\alpha\) governs the power-law exponent \(\beta\) in the degree distribution through the relationship \(\beta = 2\alpha + 1\). As \(R\) controls the average degree, we can have the estimation of \(R\) from the average degree \(\bar{k}\) as follows:
\[
R=2\log\frac{8N\alpha^2}{\bar{k}\pi(2\alpha-1)^2}.
\]
\textbf{Estimating \(r\) from the degree}\\
Given the degree of node \(i\), all nodes connected to \(i\) lie in a \(R\)-disk centered around \(i\), denoted by $D(i,R)$, within the Poincaré disk. We use \(D_p(i,R)\) for convenience to represent the portion of \(D(i,R)\) within the Poincaré \(R\)-disk. Let \(\textrm{N}(i)\) denote the set of neighbors of node \(i\); then the expectation of \(|\textrm{N}(i)|\) is:
{\fontsize{8.5}{8.5}\selectfont
\begin{equation*}
\begin{split}
\mathbb{E}(|\textrm{N}(i)|) &=\frac{N}{2\pi}\iint_{D_p(i,R)}\rho(r)drd\theta 
    =\frac{N}{2\pi}\int_{0}^{R}\rho(r)\hat{\theta_i}(r) dr.
\end{split}
\end{equation*}}\\
\(\hat{\theta_i}(r) =2\pi \) if \(r\lessapprox R-r_i\), and \(\hat{\theta_i}(r) =2\theta_i(r) \) otherwise, where \(\theta_i(r)\) is defined as follows: suppose there is a point \(x\) on the boundary of the disk \(D(i,R)\), then \(\theta_i(r) \in [0, \pi]\) is the angular difference between \(x\) and \(i\), as shown in Figure 1(a).

According to approximation (1), for large \(R\), \(r_i\), and \(r\) we have \(\theta_i(r)\approx2e^{(R-r_i-r)/2}\).
So by substituting it into the integral and keeping the leading terms we have:
\begin{equation}
\mathbb{E}(|\textrm{N}(i)|)=\frac{4N\alpha}{\pi(2\alpha-1)}e^{-r_i/2}.
\end{equation}
Note that \(|\textrm{N}(i)|\) is nearly a Poisson random variable when \(N\) is large. So the maximum likelihood estimation for \(r_i\) given the degree \(k(i)\), is
\begin{equation}
r_i=-2\log\left(\frac{\pi(2\alpha-1)}{4N\alpha}k(i)\right).
\end{equation}

\begin{figure*}[ht]
\centering
\begin{subfigure}{0.2\textwidth}
\includegraphics[width=\textwidth]{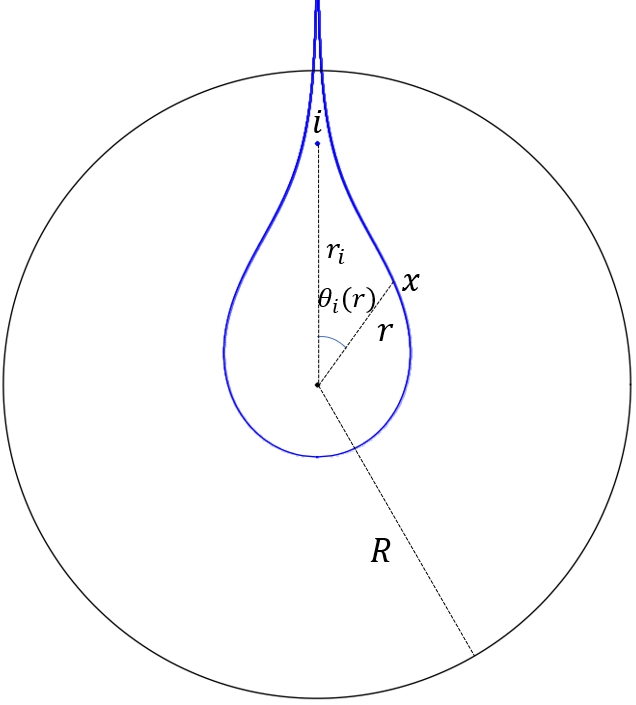}
\caption{}
\end{subfigure}\hspace{12mm}
\begin{subfigure}{0.2\textwidth}
\includegraphics[width=\textwidth]{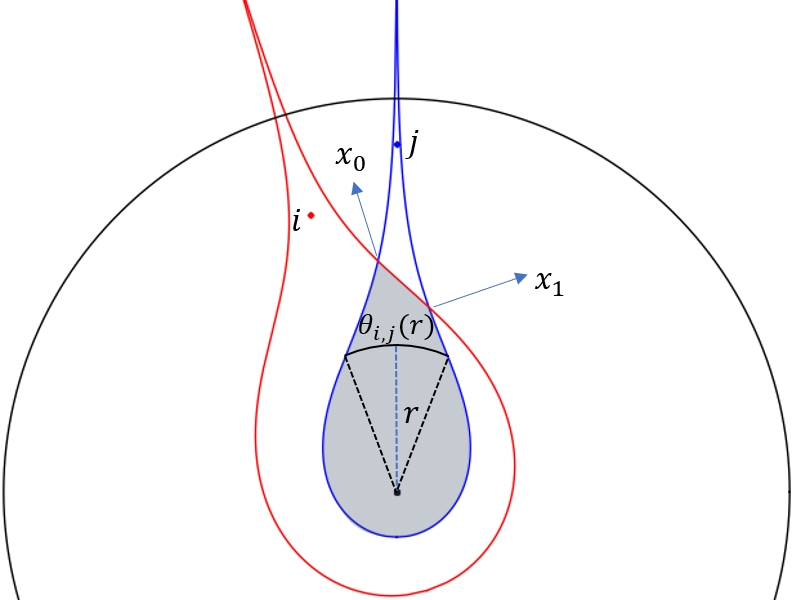}\vspace*{0.3in}
\caption{}
\end{subfigure}\hspace{12mm}
\begin{subfigure}{0.2\textwidth}
\includegraphics[width=\textwidth]{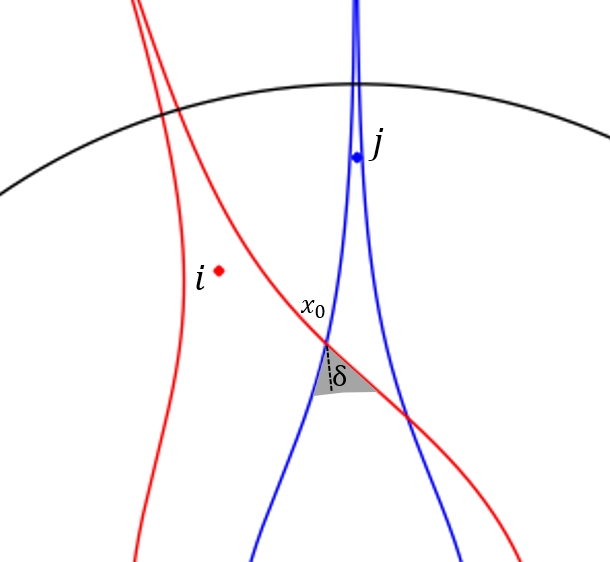}\vspace*{0.3in}
\caption{}
\end{subfigure}
\vspace{-3mm}
\caption{(a) depicts a disk with radius $R$ around node i, denoted by $D(i,R)$. (b) shows the intersection of two disks, $D(i,R)$ and $D(j,R)$. The gray area is the area we want to integrate over, where $\theta_{i,j}(r)$ is the angular section at radius $r$. (c) shows the area considered in the proof of Theorem 5.1.}
\end{figure*}

\section{Local Similarity Indices}
Most local similarity indices leverage information on the number of common neighbors and their degrees. Here are three representative examples:
\\
\textbf{Common Neighbors (CN):} 
$$S(i,j)=|\textrm{N}(i)\cap \textrm{N}(j)|.$$
This is one of the simplest techniques and has been widely used for link prediction and similarity measurement. The index counts the number of nodes connected to both \(i\) and \(j\).
\\
\textbf{Adamic/Adar Index (AA):} $$S(i,j)=\sum_{x\in \textrm{N}(i)\cap \textrm{N}(j)}\frac{1}{\log k_x},$$
where $k_x$ is the degree of a node $x$. We can see that greater weight is assigned to common
neighbors having smaller degrees. This is intuitive in the sense that two individuals are more likely to know each other if they both know a normal person rather than a celebrity.\\
\textbf{Resource Allocation (RA):} $$S(i,j)=\sum_{x\in \textrm{N}(i)\cap \textrm{N}(j)}\frac{1}{k_x}.$$
This index is very similar to the AA index, but it down weights the high degree-nodes even more. Lü and Zhou \citep{lu2011link} systematically compared various local similarity indices in real networks; their results show that the RA index has the best overall performance. Pan et al.~\citep{pan2010detecting} carried out a comparison of local indices within a similarity-based community detection algorithm; their experimental findings also indicated superior performance by the RA index.

In this section, we theoretically analyze how different weights affect link prediction performance.
\subsection{Number of Common Neighbors}
Consider two nodes \(i\) and \(j\) with $r_i\leq r_j$. Denoting the set of their common neighbors by $\textrm{CN}(i,j)$, the expectation of $|\textrm{CN}(i,j)|$ depends on the area of the intersection between two disks with a radius of $R$ within the Poincaré disk, centered at \(i\) and \(j\), respectively. Note that the boundaries of the disks $D(i,R)$ and $D(j,R)$ intersect at two distinct points, just as the circles do in Euclidean space. We denote these two intersection points by $x_0$ and $x_1$, and assume $r_{x_0}\geq r_{x_1}$; see Figure 1(b). 

The two intersection points should satisfy $\theta_i(r_{x_0})+\theta_j(r_{x_0})=\Delta \theta _{i,j}$, and $\theta_i(r_{x_1})-\theta_j(r_{x_1})=\Delta \theta _{i,j}$. Combining with the approximation (1), we have:
\begin{equation} \label{eq:2}
r_{x_0}=2\log(\frac{2e^{(R-r_i)/2}+2e^{(R-r_j)/2}}{\Delta \theta_{i,j}})
\end{equation}
\begin{equation}
r_{x_1}=2\log(\frac{2e^{(R-r_i)/2}-2e^{(R-r_j)/2}}{\Delta \theta_{i,j}}).
\end{equation}
In the same way, the angular section at radius $r$ that surrounded by the intersection shown in Figure 1(b) can be calculated as $\theta_{i,j}(r)= 2e^{(R-r_j-r)/2}+2e^{(R-r_i-r)/2}-\Delta \theta_{i,j}$ for $r_{x_1}\leq r \leq r_{x_0}$. And for $r \leq r_{x_1}$, the angular section is just $\theta_{i,j}(r)=\hat{\theta_i}(r)$, defined in Section 3.3.
Then the expectation of the number of common neighbors can be computed by:
\begin{equation*}
\begin{split}
\mathbb{E}(|\textrm{CN}(i,j)|) = \frac{N}{2\pi}\int_{0}^{\hat{r}_{x_0}}\, \rho(r)\,  \theta_{i,j}(r)  dr,
\end{split}
\end{equation*}
where $\hat{r}_{x_0}=\min(r_{x_0},R)$.\\
Note that when $r_{x_1}>R$, this expression becomes 
\begin{equation*}
\begin{split}
\mathbb{E}(|\textrm{CN}(i,j)|) =  \frac{N}{2\pi}\int_{0}^{R}\rho(y)\hat{\theta_i}(r)dr,
\end{split}
\end{equation*}
which is exactly same as the expression for $\mathbb{E}(|\textrm{N}(j)|)$ (equation 2). This corresponds to the situation where nodes \(i\) and \(j\) are very close so that the $R$-disk around \(i\) is completely surrounded by $D(j,R)$. We will mainly consider the more general case where $r_{x_1}$ is small compared with $R$.
\subsection{Estimating the Angular Difference with Lower Variance}
To estimate the latent distance of node pairs, we need to know $r_i$, $r_j$, and $\Delta \theta_{i,j}$. Since the values of $r_i$ and $r_j$ can be estimated based on the degree (see equation 3), our goal is to infer the angular difference of the pair based on the number of common neighbors.

In general, the common neighbor based indices can be written in the form of \(S=\sum_{k=1}^N w_k n_k\) where \(n_k=|\{x\in \textrm{N}(i)\cap \textrm{N}(j), \deg(x)=k\}|\) and \(w_k\) is a function of the node degree \(k\).
If \(r_i\) and \(r_j\) are known, the expectation of \(S\) can be calculated as follows:
\begin{equation*}
\begin{split}
\mathbb{E}[S]=\frac{N}{2\pi} \sum_k w_k \int_{0}^{\hat{r}_{x_0}}\, \rho(r)\, g(k|r)\, \theta_{i,j}(r)  dr,
 \end{split}
\end{equation*}
where the propagator \(g(k|r)\) is the conditional probability that a node with hidden variable \(r\) has degree \(k\). In a sparse network we assume that \(g(k|r)\) is Poissonian \citep{boguna2003class}: \(g(k|r)=\frac{e^{-\bar{k}(r)}\bar{k}(r)^k}{k!}\), where \(\bar{k}(r)=\frac{4N\alpha}{\pi(2\alpha-1)}e^{-r/2}\) is the expected degree of a node with radius \(r\), from equation (2). Using these approximations, the integral becomes:
\begin{equation*}
\begin{split}
\mathbb{E}[S]
    & \approx \sum_k w_k \frac{1}{k!}\Big[\Theta(1)f_1(r_i,r_j)\Gamma\big(k-2\alpha+1,\bar{k}(\hat{r}_{x_0})\big)\\
    &-\Theta(1)\Gamma\big(k-2\alpha,\bar{k}(\hat{r}_{x_0})\big)\Delta\theta_{i,j}\Big],
 \end{split}
\end{equation*}
where \(\Gamma\) is the upper incomplete-gamma function, and \(f_1(r_i,r_j)\) indicates a function of \(r_i\) and \(r_j\). The evaluation of the integral can be found in Appendix D.\\
Since \(\hat{r}_{x_0}=\min(r_{x_0},R)\), the expression above can be broken down into two cases depending on the value of \(r_{x_0}\). Assume $r_{x_0}\geq R$\footnote{In $r_{x_0}< R$ case, a similar result is shown in the Appendix D.}, then there is a linear relationship between \(\mathbb{E}(S)\) and \(\Delta\theta_{i,j}\): 
\begin{equation*}
\begin{split}
\mathbb{E}(S) & =\sum_k w_k \frac{1}{k!}[\Theta(1)f_1(r_i,r_j)\Gamma(k-2\alpha+1,\bar{k}(R))\\
                &-\Theta(1)\Gamma(k-2\alpha,\bar{k}(R))\Delta\theta_{i,j}].
\end{split}
\end{equation*}
Solving this equation for \(\Delta \theta_{i,j}\) gives $\hat{f}(S)=\frac{f_2(w_k,r_i,r_j,k)-S}{\Theta(1)\sum_k w_k \frac{1}{k!}\Gamma(k-2\alpha,\bar{k}(R))}$, which is an unbiased estimator of $\Delta \theta_{i,j}$ because of the linearity; \(f_2(w_k,r_i,r_j,k)\) here represents a function that depends on \(w_k,r_i,r_j\) and \(k\).

Also note that \(n_k\) can be taken to be a Poisson random variable, so the variance of $n_k$ can be calculated as:
\begin{equation*}
\begin{split}
\textrm{Var}(n_k)= \mathbb{E}(n_k)&=\frac{N}{2\pi}\int_{0}^{R}\, \rho(r)\, g(k|r)\, \theta_{i,j}(r)  dr\\
&= \frac{1}{k!}[\Theta(1)f(r_i,r_j)\Gamma(k-2\alpha+1,\bar{k}(R))\\
& \quad -\Theta(1)\Gamma(k-2\alpha,\bar{k}(R))\Delta\theta_{i,j}].
\end{split}
\end{equation*}
When \(N\) is large, the \(n_k\) should be approximately independent. Then the variance of \(\hat{f}(S)\) is:
\begin{equation*}
\begin{split}
\textrm{Var}(\hat{f}(S)) &= \frac{\sum w_k^2 \textrm{Var}(n_k)}{[\Theta(1)\sum_k w_k\frac{1}{k!} \Gamma\big(k-2\alpha,\bar{k}(R)\big)]^2} .
\end{split}
\end{equation*}
To minimize the variance of the unbiased estimator $\hat{f}(S)$, we set \(\frac{d\, \textrm{Var}(\hat{f}(S))}{d w_{k^*}}=0\), and by solving this equation (the detailed calculation is in Appendix D.), we get:\\
\begin{equation*}
\begin{split}
w_{k^*} &\propto \frac{\Gamma({k^*}-2\alpha,\bar{k}(R))}{k^*!\textrm{Var}(n_{k^*})}\\
&= \frac{\Gamma({k^*}-2\alpha,\bar{k}(R))}{\Gamma(k^*-2\alpha+1,\bar{k}(R))-\Theta(1)\Gamma(k^*-2\alpha,\bar{k}(R))\Delta \theta}\\
&\sim \frac{1}{k^*}
\end{split}
\end{equation*}
as $k^*\to \infty$, which is identical to the weight in the RA index.

An alternative way to interpret the argument above is as follows: if we consider each value of \(n_k\) individually, we can think of each \(S(n_k)=\frac{f(r_i,r_j,k)-n_k}{\Theta(1)\frac{1}{k!}\Gamma(k-2\alpha,\bar{k}(R))}\) as an unbiased estimator for the same parameter \(\Delta \theta_{i,j}\) but with different variances. Under the assumption that \(N\) is large, these estimators are nearly uncorrelated. In this context, it is possible to obtain a combined estimator with reduced variance by assigning optimal weights to each estimator.

We also found that the latent distance can be approximated as a monotonically decreasing function of the RA index (See Appendix A.). In practical situations, obtaining a precise measurement of the latent distance isn't always essential. Often, the goal is to evaluate similarities for comparative purposes. In such instances, it is convenient to utilize RA directly to perform tasks like link prediction.

In our analysis, we have justified that the RA index exhibits the least variance when estimating the distance between nodes. Additionally, in situations where we lack common neighbor information about two nodes, a naive approximation can be made by ignoring the \(\Delta\theta_{i,j}\) term. When doing so, we can approximate the distance between nodes as \(d_{ij}\approx r_i+r_j\), which can be combined with Equation (3) to show that the distance has a negative correlation with \(|\textrm{N}(i)||\textrm{N}(j)|\)—an expression that is equivalent to the preferential attachment index.

\section{Global Similarity Indices}
\subsection{The Shortest Path Index}
The shortest path index of node pair \((i,j)\) is the minimum length of a path connecting them. Besides being used as a similarity index, it is also widely used in tasks such as graph drawing, community detection, and graph embedding. 
\\
\textbf{The regular shortest-path-index}\\
For Euclidean random graphs with connection threshold \(R\), we define the graph distance as \(\hat{d}_{ij}=lR\), where \(l\) is the length of the shortest path. It is clear that the true latent distance \(d_{ij}\leq lR=\hat{d}_{ij}\) from the definition of the connection function.
An upper error bound of graph distance \(\hat{d}_{ij}\) has been derived by Arias-Castro et al.~\citep{arias2021estimation}, giving:
\\
\[
0\leq \hat{d}_{ij}-d_{ij}\leq 4(\epsilon/R)d_{ij}+R,
\]
when \(\epsilon<R/4\), where \(\epsilon=\textrm{sup}_{x\in \Omega} \textrm{min}_{i\in [N]}||x-x_i||\), indicating how dense is the points layout, and \(\Omega\) is the convex hull of all nodes. It is noteworthy that the upper error bound can not be reduced below \(R\), even in the infinitely dense setting, as \(\epsilon \rightarrow 0\).

The last error term \(R\) seems acceptable in the Euclidean graph models, where the connectivity radius can be chosen to be considerably smaller than the latent space diameter. This is not necessarily the case in hyperbolic models. In HRG, however, the connectivity radius is fixed to be \(R\), the radius of the Poincaré \(R\)-disk, which would make the upper error bound large compared with the true latent distance. In fact, since \(0\leq d_{ij}\leq 2R\) for any node pair \((i,j)\), using the standard shortest path approximation not only preserves the irreducible error term, but also introduces heavy discretization, making the index unsuitable for tasks such as link prediction. This explains the poor behavior of the shortest path index observed by Liben-Nowell \& Kleinberg~\citep{liben2003link}.
\\ 
\textbf{A modified shortest-path -index (HGD)}\\
In the hyperbolic model, given two nodes \(i\) and \(j\), the radius of any third node connecting them is constrained within a specific range depending on the values of \(r_i\), \(r_j\), and \(\Delta \theta_{i,j}\). Conversely, the radius of the intermediate node can provide bound on \(\Delta\theta_{i,j}\), given knowledge of \(r_i\) and $r_j\).

Assume we have information about the radii of all nodes. Suppose two nodes \(i\), \(j\) are connected through a path consisting of \(m_1,m_2,\ldots,m_{l-1}\). Instead of directly bounding \(d_{ij}\) by \(lR\), we can first bound the angular difference between \(i$ and \(m_1$. From equation (1) we have \(\Delta \theta_{i,m_1}\leq 2e^{R/2}e^{(-r_i-r_{m_1})/2}\) as they are connected; the bound for \(\Delta \theta_{m_1,m_2}\) can be found in the same way. We can then bound the difference in angle between \(i\) and \(m_2\) since\(\Delta \theta_{i,m_1}\leq \Delta \theta_{i,m_1}+\Delta \theta_{m_1,m_2}\), and consecutively bound the angular difference until we obtain a bound on the difference in angle between \(i\) and \(j\) by simply summing up these bounds. We can then use this bound on angle difference \(\Delta \hat{\theta}_{i,j}\) to determine that the distance between \(i\) and \(j\) is no greater than \(r_i+r_j+2\log\frac{\Delta \hat{\theta}_{i,j}}{2}\). According to this process we can define:
\[
l(i,\, j,\, P)=r_i+r_j+2\log\frac{\Delta \hat{\theta}_{i,j}}{2},
\]
where \(\Delta \hat{\theta}_{i,j}=2e^{R/2}(e^{(-r_i-r_{m_1})/2}+e^{(-r_{m_1}-r_{m_2})/2}+\cdots+e^{(-r_{m_{l-1}}-r_j)/2})\), if there exists a path \(P\) connecting \(i\) and \(j\): \(i\sim m_1 \sim\cdots\sim m_{l-1} \sim j\).
The hyperbolic graph distance (HGD) can be defined as:
\begin{equation}
    \textrm{HGD}(i,j)=\min_P l(i,j,P).  
\end{equation}
The following theorem states that for long-range distance the HGD achieves better error bounds with increasing node density.\\
\begin{theorem} Given nodes \(i\) and \(j\), and knowing \(d_{ij}\geq 2\log(e^{r_i/2}+e^{r_j/2})\),
we have:
\begin{equation}
d_{ij}\leq {\rm HGD}(i,j)\leq d_{ij}+\delta,
\end{equation}
with probability \(1-(1-p)^N\), where \(p=\frac{1}{2\pi}e^{\alpha(r_{x_0}-R)}\Delta \theta_{i,j} (\frac{1}{2\alpha-1}-\frac{2\alpha}{2\alpha-1}e^{-(\alpha-\frac{1}{2})\delta}+e^{-\alpha \delta})\).
\end{theorem}
The \(p\) above is the probability that a point falls in the area depicted in Figure 1(c). Full proof of the theorem can be found in Appendix B. This theorem shows the consistency of HGD for estimating distance when \(d_{ij}\) is relatively large. In the case when \(d_{ij}< 2\log(e^{r_i/2}+e^{r_j/2})\), causing \(r_{x_0}>R\), a similar derivation gives us an upper bound of \(d_{ij}+\delta+(r_{x_0}-R)\). Although we now have an extra term \(r_{x_0}-R\), this term would be small compared to \(R$ as long as \(r_{x_0}\) is not too big, i.e., \(i$ and \(j\) are not too close.
\\
\noindent \textbf{Algorithms to calculate HGD}\\
Given a graph, we can estimate the latent radius of node \(i\) based on its degree \(k(i)\) according to equation (3). After plugging in, the HGD index can be written neatly as:
{\fontsize{8.5}{8.5}\selectfont
\begin{equation}
\label{eq:HGD}
\textrm{HGD$(i,j)$}=\min_P \big(R+2\log(\frac{k(i)k(m_1)+...+k(m_{l-1})k(j)}{k(i)k(j)})\big).
\end{equation}}
We can see that in addition to the path length, HGD assigns weights according to the degree of nodes along the path. This adjustment not only makes the index less discrete but also accounts for the relative importance of nodes based on their connectivity.

Many algorithms, such as Dijkstra's, have been proposed for calculating the shortest path between nodes, both for unweighted and weighted networks. To compute the hyperbolic graph distance, a simple modification can be applied to transform the problem into finding the shortest weighted paths between node pairs: Initially, for every node pair \((m_1,m_2)\) we assign the edge weight to be \(k(m_1)k(m_2)\). Then, by treating the network as a weighted network, we can determine the shortest weighted path between \(i\) and \(j\). Dividing this value by \(k(i)k(j)\) yields the hyperbolic graph distance between node \(i\) and \(j\).
\subsection{The Number of \(l\)-hop Paths and the Katz Index}
There is usually more than one path between node \(i\) and \(j\). That is to say, by using the shortest path as the link prediction index, we omit some other information about the similarity. An index that tries to capture this information is the number of \(l\)-hop paths between two nodes, where an \(l\)-hop path consists of \(l\) edges between two nodes:
\[
l\textrm{-hop}(i,j)=|\text{path}^{l}_{i,j}|.
\]
Inspired by previous discussion, we can define a weighted variant of the number of \(l\)-hop paths:
\[
l\textrm{-hop*}(i,j)=\sum \text{path}^{*l}_{i,j},
\]
where \(\text{path}^{*l}_{i,j}\) means that the path is weighted according to the degrees of the nodes connecting node \(i\) and node \(j\). We assume the weight is \(w_k\). Through an analysis similar to the one in Section 4, it can be shown that choosing \(w_k\) to be \(\frac{1}{k}\) leads to a smaller error bound when estimating the angular difference (analyzed in Appendix C). This modification can be also extended to Katz index \citep{katz1953new}, which is a emsemble of \(l\)-hop paths indices. For detailed definitions of Katz and its variants, refer to Appendix E.

\section{Experimental Analysis}
\subsection{Numerical Experiments on HRG}
In this section, we verify our analysis by running experiments on HRGs. We generated 6 hyperbolic random graphs with different choices of $N$, $\alpha$ and $R$. We randomly sample 10000 pairs of nodes, and for each index calculate the Spearman's correlation coefficient with respect to the true latent distance as an evaluation metric. For link prediction, we following the standard routine: we randomly remove $10\%$ of the total links and use indices on the remaining graph to predict missing links. The performance is measured by average precision. In Appendix F we provide experiment details and more tests on different configurations of HRG. Appendix H explains the evaluation metrics used in these experiments.\\
\textbf{Results Analysis}\\
Table 1 presents the results. Among the three local indices (CN, AA, and RA), RA shows the strongest correlation with latent distances, consistent with both our theoretical analysis and prior empirical findings \citep{lu2011link}. For global indices, the weighted versions consistently outperform their unweighted counterparts across all networks. Notably, the HGD index performs exceptionally well, demonstrating a strong correlation with the true latent distances.

The right part of table 1 displays the precision scores for link prediction. Much like the correlation test, the weighted version of each index outperforms the unweighted version. Notice that the standard graph distance proves to be nearly useless in this task, yielding scores close to 0. Sensitivity to missing links also impacts the HGD index, causing the latter to no longer be the best-performing index in this task. It still performs considerably better, however, than the standard graph distance (GD). This also confirms that the shortest path indices might be a good choice for a similarity measure, but a poor choice for link prediction.
\begin{table*}[t]
\centering
\resizebox{1\textwidth}{!}{
\begin{tabular}{ccccccccccccc} 
\hline
        & \multicolumn{6}{c}{Correlation}                                                                                 & \multicolumn{6}{c}{Link Prediction Precision}                                                                        \\ 
\cmidrule(lr){1-7}
\cmidrule(lr){8-13}
        & 1                & 2                & 3                & 4                & 5                & 6                & 1               & 2               & 3               & 4               & 5               & 6                \\ 
\cmidrule(lr){1-7}
\cmidrule(lr){8-13}
CN      & -0.8996          & -0.5340          & -0.8043          & -0.4279          & -0.7079          & -0.3329          & 0.5430          & 0.3413          & 0.4950          & 0.2907          & 0.6286          & 0.4077           \\
AA      & -0.9070          & -0.5358          & -0.8172          & \textit{-0.4286} & -0.7124          & \textit{-0.3331} & 0.6897          & 0.5273          & 0.6679          & 0.4818          & 0.7461          & 0.5774           \\
RA      & \textit{-0.9098} & \textit{-0.5360} & \textit{-0.8217} & \textit{-0.4286} & \textit{-0.7137} & \textit{-0.3331} & 0.8041          & 0.5879          & 0.7596          & 0.5108          & \textbf{0.8141} & \textbf{0.6331}  \\ 
\cmidrule(lr){1-7}
\cmidrule(lr){8-13}
GD      & 0.6990           & 0.5865           & 0.7010           & 0.5741           & 0.6957           & 0.4608           & 0.0092          & 0.0162          & 0.0058          & 0.0189          & 0.0056          & 0.0210           \\
HGD     & \textbf{0.9704}  & \textbf{0.8396}  & \textbf{0.9533}  & \textbf{0.7959}  & \textbf{0.9553}  & \textbf{0.8102}  & 0.6463          & 0.3652          & 0.5346          & 0.3053          & 0.4435          & 0.2914           \\ 
\cmidrule(lr){1-7}
\cmidrule(lr){8-13}
3-hop   & -0.9058          & -0.6054          & -0.8644          & -0.5405          & -0.8267          & -0.4814          & 0.3334          & 0.1779          & 0.2514          & 0.1442          & 0.2940          & 0.1823           \\
3-hop*  & \textit{-0.9580} & \textit{-0.6481} & \textit{-0.9263} & \textit{-0.5614} & \textit{-0.8736} & \textit{-0.4891} & 0.7186          & 0.4739          & 0.6922          & 0.3923          & 0.7312          & 0.5046           \\ 
\cmidrule(lr){1-7}
\cmidrule(lr){8-13}
Katz-3  & -0.9252          & -0.6094          & -0.8723          & -0.5421          & -0.8267          & -0.4815          & 0.5498          & 0.3508          & 0.5038          & 0.3031          & 0.6364          & 0.4110           \\
Katz*-3 & -0.9492          & -0.6455          & -0.9098          & -0.5619          & -0.8678          & -0.4890          & \textbf{0.8068} & \textbf{0.5902} & \textbf{0.7632} & \textbf{0.5143} & \textbf{0.8141} & \textbf{0.6334}  \\
Katz    & -0.9246          & -0.6137          & -0.8672          & -0.5806          & -0.8241          & -0.4776          & 0.5491          & 0.3511          & 0.5036          & 0.3036          & 0.6361          & 0.4111           \\
Katz*   & \textit{-0.9586} & \textit{-0.6976} & \textit{-0.9288} & \textit{-0.6442} & \textit{-0.8922} & \textit{-0.5362} & \textbf{0.8065} & 0.5896          & 0.7619          & \textbf{0.5139} & \textbf{0.8141} & \textbf{0.6334}  \\
\hline
\end{tabular}
}
\label{syn}
\caption{Spearman's correlation and link prediction precision in HRGs. Indices weighted by $1/k$ are marked with an asterisk (*). The Katz-3 index represents a localized version of the Katz index, considering paths of up to 3 hops.}
\end{table*}

\subsection{Applications}
In the previous section, we analyzed similarity heuristics in synthetic networks. In principle, these modified heuristics can serve as direct substitutes in real-world network tasks where node similarity needs to be assessed. Here, we briefly discuss a few illustrative use cases.\\
\textbf{Link Prediction}\\
As machine learning evolves, more sophisticated link prediction methods are proposed. However, research \citep{zhang2021labeling} shows that simple heuristics like CN and RA can sometimes outperform advanced methods like Graph Autoencoders or Message Passing Neural Networks by large margins. In table 2 and appendix G we test the link prediction heuristics we have discussed on real networks to compare their performance.

The results in Table 2 demonstrate that the behavior of these indices on real networks is consistent with what we observe in synthetic hyperbolic graphs: weighted versions consistently outperform their unweighted counterparts and are often among the top-performing methods.
As in the synthetic case, shortest-path-based indices such as GD and HGD are sensitive to missing links, which affects their predictive accuracy. We therefore recommend HGD primarily as a similarity measure for tasks such as clustering and embedding—particularly in settings where the full graph is preserved—since in such cases, HGD can provide more informative estimates of similarity than standard graph distance.\\
\textbf{Graph Embedding using HGD}\\
In representation learning on graphs, the goal is to embed nodes as low-dimensional vectors. Hyperbolic space has shown strong potential for representing hierarchical or tree-structured data \citep{nickel2017poincare, sala2018representation}, and many existing methods rely on estimating pairwise node similarities. Hyperbolic embedding approaches such as hyperbolic ordinal embedding \citep{suzuki2019hyperbolic} and h-MDS \citep{sala2018representation} are typically based on pairwise distances, which are often approximated using graph distance. We propose using Hyperbolic Graph Distance (HGD) as a more accurate substitute for graph distance when estimating latent distances for embedding. Table 3 demonstrates that a simple substitution of graph distance with HGD leads to reduced embedding distortion. Appendix I provides further details on the methodology and evaluation metrics.

\begin{table}[!htb]
    \centering
    
    \medskip
\vspace{-3mm}
\resizebox{1\columnwidth}{!}{
\begin{tabular}{llllll} 
\hline
       & Karate          & USair97         & C.elegans        & Harvard  & Cora           \\ 
\hline
CN     & 0.1122          & 0.3363          & 0.0663          & 0.1111  & 0.0131         \\
AA     & 0.1349          & 0.3728          & 0.1312          & 0.1227  & 0.0229         \\
RA     & 0.1317          & 0.4244          & 0.1756          & \textbf{0.1552}  & 0.0216         \\ 
\hline
GD     & 0.0209          & 0.0109          & 0.0052          & 0.0087 & 0.0037          \\
HGD    & 0.1269          & 0.3225          & 0.0655          & 0.0267   & 0.0028        \\ 
\hline
3-hop   & 0.0883          & 0.3287          & 0.0821          & 0.0559  & 0.0082         \\
3-hop*  & \textbf{0.1643} & 0.4164          & 0.1243          & 0.0837   & 0.0248        \\ 
\hline
Katz-3  & 0.1583          & 0.3476          & 0.0869          & 0.1069 & 0.0185          \\
Katz*-3 & 0.1523          & \textbf{0.4460} & \textbf{0.1774} & \textbf{0.1551}   & 0.0273        \\
Katz   & 0.1513          & 0.3470          & 0.0782          & 0.1058  & 0.0185         \\
Katz*  & 0.1514          & 0.4363          & 0.1766          & \textbf{0.1551} & \textbf{0.0285} \\
\hline
\end{tabular}
}
\caption{Link prediction precision on selected real-world networks. Results on more datasets can be found in Appendix G.}
\label{tab:link-prediction}
\end{table}

\begin{table}[!htb]
    \centering
    
    \medskip
\vspace{-3mm}
\begin{tabular}{ccccc}
\hline
          & \multicolumn{2}{c}{h-MDS} & \multicolumn{2}{c}{h-MDS*}  \\ 
\cmidrule(lr){2-3}
\cmidrule(lr){4-5}
Dimension =      & 2      & 10               & 2      & 10                 \\ 
\cmidrule(lr){2-3}
\cmidrule(lr){4-5}
Bal. Tree & 0.2362 & 0.6055           & 0.2534 & \textbf{0.6368}    \\
Phy. Tree & 0.1186 & 0.1283           & 0.5395 & \textbf{0.6458}    \\
CS PhDs   & 0.1921 & 0.2586           & 0.2374 & \textbf{0.3789}    \\
WordNet   & 0.2255 & 0.2969           & 0.2353 & \textbf{0.3333}    \\
\hline
\end{tabular}
\caption{The MAP (mean average precision) of h-MDS, where h-MDS* implies using HGD instead of graph distance.}
\label{tab:embedding-map}
\end{table}
\section{Conclusion}
In this paper, we analyze popular link heuristics in the context of the hyperbolic graph model. We provide a theoretical justification for the effectiveness of weighting common neighbors based on their degrees. Our analysis further extends to global similarity measures, such as the number of $l$-hop paths, yielding similar insights. Additionally, we introduce the Hyperbolic Graph Distance (HGD), a modified version of the shortest path index, which demonstrates lower error in both theoretical analysis and empirical evaluations.

We believe that hyperbolic latent space models offer valuable insights into the structure of complex networks, as they provide a principled explanation for many empirical phenomena. However, real-world networks often go beyond simple graphs, highlighting the need to incorporate further factors in future studies.

\clearpage
\appendix
\input{anonymous-submission-latex-2026_appendix}

\bibliography{refs_AAAI}

\end{document}

%% file: anonymous-submission-latex-2026_appendix.tex
\section{Estimation of the Latent Distance with Weighted Common Neighbors}
Subject to some approximations, we know that selecting the weight $w_{k}=\frac{1}{k}$ gives a lower variance in terms of estimating the angle difference. Now we can establish a relationship between the latent distance and the index $S^*=\sum_{k=1}^N \frac{1}{k}n_k$.
Different treatment is required for the cases where $r_{x_0}< R$ and $r_{x_0}\geq R$ in our analysis. This is due to the fact that the change in the upper bound of the integration leads to different expressions for these two scenarios.
\\
\textbf{Case 1:} $r_{x_0}< R$\\
Note that $r_{x_0}< R$ implies that $\Delta\theta_{i,j} > (2e^{-r_i/2}+2e^{-r_j/2})$. While it is hard to compute the expectation exactly, obtaining an approximation is relatively straightforward:
\begin{align*}
\mathbb{E}(S^*) & = \frac{N}{2\pi} \sum_k w_k^* \int_{0}^{r_{x_0}}\, \rho(r)\, g(k|r)\, \theta_{i,j}(r)  dr\\
&=\frac{N}{2\pi}\int_{0}^{r_{x_0}}\sum_k  w_k^* g(k|r)\rho(r)\theta_{i,j}(r)  dr\\
&\approx \frac{N}{2\pi}\int_{0}^{r_{x_0}} \frac{\rho(r)\theta_{i,j}(r)}{\bar{k}(r)} dr\\
&= \frac{2\alpha-1}{2\alpha}\exp(\frac{R}{2}-\frac{r_j}{2}-\alpha r_i)(\frac{2}{\Delta \theta_{i,j}})^{2\alpha}.
\end{align*}
By taking inverse of above expression we have an estimation for $\Delta\theta_{i,j}$ in term of $S^*$, plugging which into equation(1) in the main paper gives the estimation of distance:
{\fontsize{8.5}{8.5}\selectfont
\[
\hat{d_{ij}}=f(S^*)= \frac{1}{2\alpha}R+(1-\frac{1}{2\alpha})r_j+\frac{1}{\alpha}ln(\frac{2\alpha-1}{2\alpha})-\frac{1}{\alpha}\log S^*.
\]}
Note that in the unweighted case, by taking the inverse of the given expectation, we obtain a maximum likelihood estimator (MLE) because the number of common neighbors can be considered as a Poisson random variable. It is known that the MLE for the parameter of a Poisson random variable is the expectation itself. However, in the weighted common neighbor case, the estimator we get is not guaranteed to be MLE. 

Another observation is that we can effectively treat the second term involving $r_j$ as a constant, thus the approximation simplifies into a single-variable function $f(S^*)=\Theta(1)-\frac{1}{\alpha}\ln S^*$. This simplification is reasonable because of the assumption that $r_j>r_i$, and the characteristics of the radius distribution, which indicate that $r_j$ is likely to be close to $R$, and the variance of $r_j$ should be small.\\
\textbf{Case 2:} $r_{x_0} \geq R$\\
Combined with the assumption $r_{x_1}< R$, we have the inequality: 
\[
(2e^{-r_i/2}-2e^{-r_j/2})<\Delta \theta_{i,j} \leq (2e^{-r_i/2}+2e^{-r_j/2}).
\]
Since $r_j$ is the bigger radius, close to R, this implies that $\Delta \theta_{i,j}$ lies in a small interval around $2e^{-r_i/2}$. By plugging into $d\approx r_i+r_j+2ln\frac{\Delta \theta_{i,j}}{2}$, we know that in this case $d$ is very close to $r_j$.\\
By computing the expectation as follows:
{\fontsize{8.5}{8.5}\selectfont
\begin{equation*}
\begin{split}
\mathbb{E}(S^*) & \approx\frac{2\alpha-1}{8}(\frac{1}{\alpha}(2e^{(R-r_j)/2}+2e^{(R-r_i)/2})-\frac{1}{\alpha+\frac{1}{2}}e^{\frac{R}{2}}\Delta \theta_{i,j})\\
&=\frac{2\alpha-1}{8}(\frac{1}{\alpha}(2e^{(R-r_j)/2}+2e^{(R-r_i)/2})\\
& \quad -\frac{2}{\alpha+\frac{1}{2}}e^{(R+d-r_i-r_j)/2}).\\
\end{split}
\end{equation*}}
We can then simplify the expression by the first order Taylor's expansion of $r_i$ and $r_j$ around $R$. Then the above expression can be written as:
\begin{equation*}
\begin{split}
\mathbb{E}(S^*) &\approx\frac{2\alpha-1}{4}\big(\frac{1}{\alpha}(1+\frac{R-r_j}{2}+1+\frac{R-r_i}{2})\\
& \quad-\frac{1}{\alpha+\frac{1}{2}}(\frac{R+d-r_i-r_j}{2}+1)\big).
\end{split}
\end{equation*}
By taking inverse of the function we have the estimator:
{\fontsize{8.5}{8.5}\selectfont
\[
\hat{d_{ij}}=f(S^*)= \frac{\alpha+1/2}{\alpha}(R+2)-\frac{1}{2\alpha}(r_i+r_j)-\frac{4(2\alpha+1)}{2\alpha-1}S^*.
\]}
To further simplify, we can also treat the second term as a constant considering the relatively small coefficients and the radial distribution. Consequently, we have a negative correlation between the latent distance $d$ and $S^*$:
\[
f(S^*)\approx \Theta(1) - \frac{4(2\alpha+1)}{2\alpha-1}S^*.
\]

\section{Proof of Theorem 5.1}
\begin{proof}
From the definition we can directly have $d_{ij}\leq \textrm{HGD}(i,j)$. Next we look at the other side of the inequality.

From the expression of $r_{x_0}$ (equation 3), we know that the condition $d_{ij}\geq 2\log(e^{r_i/2}+e^{r_j/2})$ implies a restriction on $r_{x_0}$, namely $r_{x_0}\leq R$.
Let's consider the region defined as the intersection of $D(i,R)$, $D(j,R)$, and $\{x|r_{x_0}-\delta\leq r_x\leq r_{x_0}\}$. In simpler terms, this region represents the overlapping area between two disks, with their radii falling within the specific interval $[r_{x_0}-\delta,r_{x_0}]$. The probability of any given node falling into this region can be expressed as:
{\fontsize{8.5}{8.5}\selectfont
\begin{equation*}
\begin{split}
p &= \frac{1}{2\pi}\int_{r_{x_0}-\delta}^{r_{x_0}} \rho(r) \theta_{i,j} dr\\
&= \frac{1}{2\pi}\int_{r_{x_0}-\delta}^{r_{x_0}}\alpha e^{\alpha(r-R)}(2e^{(R-r_j-r)/2}\\
 & \quad \quad \quad \quad \quad \quad \quad \quad \quad +2e^{(R-r_i-r)/2}-\Delta \theta_{i,j}) dr\\
&=\frac{1}{2\pi}e^{\alpha(r_{x_0}-R)}\Delta \theta_{i,j} (\frac{1}{2\alpha-1}-\frac{2\alpha}{2\alpha-1}e^{-(\alpha-\frac{1}{2})\delta}+e^{-\alpha \delta}).
\end{split}
\end{equation*}}
Thus probability that there exist at least one node in this region is $1-(1-p)^N$. Any node $m$ landing within this region must satisfy $r_m\geq r_{x_0}-\delta$ and $m\in \textrm{N}(i)\cap \textrm{N}(j)$. Then
\begin{equation*}
\begin{split}
\textrm{HGD$(i,j)$} &\leq r_i+r_j \\ &\quad +2\log\big(e^{R/2}(e^{(-r_i-r_m)/2}+e^{(-r_m-r_j)/2})\big)\\ 
&\leq d_{ij}+\delta,
\end{split}
\end{equation*}
with the same probability.
\end{proof}

\section{The Weighted Number of $l$-hop Paths}
The simplest case of the number of $l$-hop paths would be the number of 3-hop paths, since the number of 2-hop paths is equivalent to the number of common neighbors, which has been discussed in section 4.\\
\subsection{A Weighted Scheme of the 3-hop Paths}
Consider a weighted variant of the number of $3$-hop paths:
\[
S^*_3(i,j)=\sum \text{path}^{*3}_{i,j}.
\]
Here, $\text{paths}^{*3}_{i,j}$ signifies that the path is weighted according to the degrees of the nodes connecting node $i$ and node $j$. In other words, if there is a $3$-hop path that connects node $i$ and node $j$, that is, $i\sim m_1\sim m_2\sim j$, then $\text{paths}^{*3}_{i,j}=w(\textrm{deg}(m_1))w(\textrm{deg}(m_2))$. The function $w(\textrm{deg}(m))$ relies on the degree of node $m$, and we may represent this notation more concisely as $w_k$.\\
Computing the exact expectation of such index analytically is quite intricate and can make the analysis of this index difficult.  In section 4 and appendix D.2, we discuss that, for sufficiently large $k$, there exists an approximately linear relationship between $\textrm{N}_k(i)$ and $\Delta \theta_{i,j}$. In order to streamline our analysis, we focus on the cases where $k$ is large. Specifically, we define $S^*_3(i,j,k_0)$ to be the sum of weighted 3-hop paths, where the degree of intermediate nodes is no less than $k_0$.

For the expectation, consider the form $S^*_3(i,j,k_0)=\sum_{k=k_0}^N w_k S^*_{k,2}$, where\\
$S^*_{k,2}=\sum_{\substack{x\in \textrm{N}(i)\\ \textrm{deg}(x)=k}}\textrm{CN}^*(x,j)$, and $\textrm{CN}^*(x,j)$ indicate the weighted number of common neighbors between $x$ and $j$. 
In the following process we only consider the nodes with degree no less than $k_0$. 
The expectation of $S^*_3(i,j,k_0)$ is:
\begin{equation*}
\begin{split}
\mathbb{E}(S^*_3(i,j,k_0))&=\frac{N}{2\pi}\sum_{k=k_0}^N w_k \iint_{x \in D_p(i,R)}\rho(r_{x})g(k|r_{x})\\
&\quad \frac{N}{2\pi}\sum_{l=k_0}^N w_l \iint_{D_p(x,R)\cap D_p(j,R)}\rho(r)g(l|r).
\end{split}
\end{equation*}
\\
Similar to section 4, under the case where $k_0$ is big, the last integral can be approximated as:
{\fontsize{8.5}{8.5}\selectfont
\begin{equation*}
\begin{split}
&\frac{N}{2\pi}\sum_{l=k_0}^N w_l \iint_{D_p(x,R)\cap D_p(j,R)}\rho(r)g(l|r)\\
& \approx \sum_l w_l \frac{1}{l!} \big(\Theta(1)f(r_i,r_{x})\Gamma(l-2\alpha+1)-\Theta(1)\Gamma(l-2\alpha)\Delta\theta_{x,j}\big).
\end{split}
\end{equation*}
}
Then the expectation becomes:
{\fontsize{8.5}{8.5}\selectfont
\begin{equation*}
\begin{split}
\mathbb{E}(S^*_3(i,j,k_0))&=\Theta(1)\sum_{k=k_0}^N w_k \iint_{x \in D_p(i,R)}\rho(r_{x})g(k|r_{x})*\\
&\quad \sum_{l=k_0}^N w_l\frac{1}{l!} \big(\Theta(1)f(r_i,r_{x})\Gamma(l-2\alpha+1)-\\
&\quad \Theta(1)\Gamma(l-2\alpha)\Delta\theta_{x,j}\big)\\
            &\approx \Theta(1)[\sum_k w_k \frac{1}{k!} \Theta(1)f(r_i)\Gamma(k-2\alpha+1)]*\\
            &\quad [\sum_l w_l \frac{1}{l!} \big(\Theta(1)f(r_i,r_{\xi})\Gamma(l-2\alpha+1)-\\
            &\quad \Theta(1)\Gamma(l-2\alpha)\Delta\theta_{i,j}\big)],
\end{split}
\end{equation*}}
\\
where $\xi$ is a imaginary node lies in $D_p(i,R)$ according to the intermediate value theorem; $f(r_i), f(r_i,r_{\xi})$ indicate functions depend on $r_i$, and $r_i,r_{\xi}$ respectively, which we can ignore later. Note that $\xi$ does not depend on $\Delta\theta_{i,j}$, so as we did in section 4, we can construct an unbiased estimator of $\Delta\theta_{i,j}$, in terms of $S^*_3(i,j,k_0)$:
{\fontsize{8.5}{8.5}\selectfont
\begin{equation}
\begin{split}
\hat{f}&(S^*_3(i,j,k_0))=\Theta(1)\frac{\sum_k w_k \frac{1}{k!} (f(r_i,r_{\xi})\Gamma(k-2\alpha+1)}{\Gamma(k-2\alpha)}-\\
&\Theta(1)\frac{S^*_3(i,j,k_0)}{[\sum_l w_l \frac{1}{l!}\Gamma(l-2\alpha)]\cdot [\sum_k w_k \frac{1}{k!}f(r_i)\Gamma(k-2\alpha+1)]}. 
\end{split} \tag{C.1.1}
\end{equation}
}
Then the variance of the estimator is:
{\fontsize{8.5}{8.5}\selectfont
\begin{equation}
    \begin{aligned}
&\textrm{Var}\bigl(\hat{f}(S^*_3(i,j,k_0))\bigl) =\\
&\Theta(1)\frac{\textrm{Var}(S^*_3(i,j,k_0))}{[\sum_k w_k \frac{1}{k!}\Gamma(k-2\alpha)]^2\cdot [\sum_k w_k \frac{1}{k!}f(r_i)\Gamma(k-2\alpha+1)]^2}\\.
    \end{aligned}\tag{C.1.2}
\end{equation}
}
Now we look at the variance of $S^*_3(i,j,k_0)$. Denote $\textrm{N}^*(i)$ as the weighted number of neighbors of node $i$.  we have the upper bound as shown in C.1.3 (presented in table 1). In the equation C.1.3, at (1) we assume that $\textrm{N}^*(i)$ and $\textrm{N}^*(j)$ are independent. And at (2) we only keep the last two terms of (1) which dominant the expression when $N$ is big.\\
\begin{table*}
\centering
\begin{minipage}{0.75\textwidth}
\begin{align*}
    \textrm{Var}(S^*_3(i,j,k_0)) &=\textrm{Var} \big( \sum_l w_l \sum_{\substack{x_1\in \textrm{N}(i)\\deg(x_1)=l}}\textrm{CN}^*(x_1,j)\big)\\
    & \leq \textrm{Var} \big( \sum_l w_l \sum_{\substack{x_1\in \textrm{N}(i)\\deg(x_1)=l}}\textrm{N}^*(j)\big)\\
    & = \textrm{Var}(\textrm{N}^*(i)\cdot \textrm{N}^*(j)) \\
    &= \textrm{Var}(\textrm{N}^*(i))*\textrm{Var}(\textrm{N}^*(j))+\textrm{Var}(\textrm{N}^*(i))\mathbb{E}[\textrm{N}^*(j)]^2+\textrm{Var}(\textrm{N}^*(j))\mathbb{E}[\textrm{N}^*(i)]^2 &&\text{(1)}\\
    &\approx [\sum_k w_k^2 \frac{1}{k!} \Theta(1)f(r_i)\Gamma(k-2\alpha+1)]\cdot [\sum_k w_k \frac{1}{k!} \Theta(1)f(r_j)\Gamma(k-2\alpha+1)]^2  &&\text{(2)}\\
    &\quad + [\sum_k w_k^2 \frac{1}{k!} \Theta(1)f(r_j)\Gamma(k-2\alpha+1)]\cdot [\sum_k w_k \frac{1}{k!} \Theta(1)f(r_i)\Gamma(k-2\alpha+1)]^2.
\end{align*}
\medskip
\hrule
\end{minipage}
\caption{The Equation of C.1.3}
\end{table*}
Plugging above expression into (C.1.2) we have the upper bound of $\textrm{Var}\bigl(\hat{f}(S^*_3(i,j,k_0))\bigl)$. 
{\fontsize{8.5}{8.5}\selectfont
\begin{equation*}
\begin{split}
\textrm{Var}\bigl(\hat{f}(S^*_3(i,j,k_0))\bigl) \leq \Theta(1) \frac{\sum_k w_k^2 \frac{1}{k!}\Gamma(k-2\alpha+1)}{[\sum_k w_k \frac{1}{k!}\Gamma(k-2\alpha)]^2}\\.
\end{split}
\end{equation*}
}
Minimizing the bound in terms of $w_k$ gives us the optimal weight $w_k^*\sim\frac{1}{k}$.\\
The variance of the estimator can not be compute explicitly and hence hard to optimize. As a relaxation, we instead minimize the upper bound of the estimator to obtain a more tight error bound for the new estimator. While $S^*_3(i,j,k_0)$ constitutes a portion to the index $S^*_3(i,j)$, empirical results in section 6 show that the weight scheme improves the overall performance of the total estimator $S^*_3(i,j)$.\\
\subsection{Number of $l$-hop Paths}
In the same way we can extend the analysis to $l$-hop cases. Define a weighted variant of the number of $l$-hop paths:
\[
S^*_l(i,j)=\sum \text{path}^{*l}_{i,j}.
\]
As before, $\text{path}^{*l}_{i,j}$ means that the path is weighed according to the degrees of the nodes connecting node $i$ and node $j$. We assume the weight is unknown and denote as $w_k$.
\\
Here we also use the notation $S^*_l(i,j,k_0)$ to denote the case where the degree of intermediate nodes is no less than $k_0$.
Then the expectation of the weighted number of $l$-hop paths given nodes $i$ and $j$ is the expression shown in Table 2, equation C.2.1.

\begin{table*}
\centering
\begin{minipage}{0.75\textwidth}
\begin{align*}
    \mathbb{E}(S_l^*(i,j,k_0))&=\frac{N}{2\pi}\sum_{k=k_0}^N w_k \iint_{x_1 \in D_1}\rho(r_{x_1})g(k|r_{x_1})\frac{N}{2\pi}\sum_{k=k_0}^N w_k \iint_{x_2 \in D_2|x_1}\rho(r_{x_2})g(k|r_{x_2})...\\
&...\frac{N}{2\pi}\sum_{k=k_0}^N w_k \iint_{x_{l-1} \in D_{l-1}|x_{l-2}}\rho(r_{x_{l-1}})g(k|r_{x_{l-1}}).\tag{C.2.1}
\end{align*}

\begin{align*}
\mathbb{E}(S_l^*(i,j,k_0))&\approx \Theta(1)F(w_k,r_i,r_j)-\\
&\Theta(1)\big(\prod_{m=1}^{l-2}\sum_{k=k_0}^N w_k\frac{1}{k!}e^{(R-r_{\xi_m})/2} \Gamma(k-2\alpha+1)\big)\big(\sum_{k=k_0}^N w_k\frac{1}{k!} \Gamma(k-2\alpha)\big)\Delta\theta_{i,j}. \tag{C.2.2}
\end{align*}

\begin{align*}
\hat{f}(S_l^*(i,j,k_0))=\frac{\Theta(1)F(w_k,r_i,r_j)-S_l^*(i,j,k_0)}{\Theta(1)\big(\prod_{m=1}^{l-2}\sum_{k=k_0}^N w_k\frac{1}{k!}e^{(R-r_{\xi_m})/2} \Gamma(k-2\alpha+1)\big)\big(\sum_{k=k_0}^N w_k\frac{1}{k!} \Gamma(k-2\alpha)\big)}. \tag{C.2.3}
\end{align*}

\begin{align*}
\textrm{Var}[\hat{f}(S_l^*(i,j,k_0))] &\leq \Theta(1)\frac{\Big(\prod_{n=0}^{l-3}\mathbb{E} [\textrm{N}^*(\xi_n)]^2 \Big)\cdot \textrm{Var}\big(\textrm{N}^*(\xi_{l-2})\big)}{\big(\prod_{m=1}^{l-2}\sum_{k=k_0}^N w_k\frac{1}{k!}e^{(R-r_{\xi_m})/2}) \Gamma(k-2\alpha+1)\big)^2\big(\sum_{k=k_0}^N w_k\frac{1}{k!} \Gamma(k-2\alpha)\big)^2}.\\ \tag{C.2.4}
\end{align*} 

\begin{align*}
&\textrm{Var}[\hat{f}(S_l^*(i,j,k_0))]\\
& \leq \Theta(1)\frac{\big(\prod_{n=0}^{l-3}\sum_{k=k_0}^N w_k\frac{1}{k!}e^{(R-r_{\xi_n})/2} \Gamma(k-2\alpha+1)\big)^2\cdot \big(\sum_{k=k_0}^N w_k^2 \frac{1}{k!}e^{(R-r_{\xi_{l-2}})/2}\Gamma(k-2\alpha+1)\big)}{\big(\prod_{m=1}^{l-2}\sum_{k=k_0}^N w_k\frac{1}{k!}e^{(R-r_{\xi_m})/2} \Gamma(k-2\alpha+1)\big)^2\big(\sum_{k=k_0}^N w_k\frac{1}{k!} \Gamma(k-2\alpha)\big)^2} &&(1)\\
&\leq \Theta(1)\frac{\big(\prod_{n=0}^{l-3}\sum_{k=k_0}^N w_k\frac{1}{k!}e^{R/2} \Gamma(k-2\alpha+1)\big)^2\cdot \big(\sum_{k=k_0}^N w_k^2 \frac{1}{k!}e^{R/2}\Gamma(k-2\alpha+1)\big)}{\big(\prod_{m=1}^{l-2}\sum_{k=k_0}^N w_k\frac{1}{k!}e^{(R-R)/2} \Gamma(k-2\alpha+1)\big)^2\big(\sum_{k=k_0}^N w_k\frac{1}{k!} \Gamma(k-2\alpha)\big)^2} &&(2)\\
&=\Theta(1)\frac{\big(\sum_{k=k_0}^N w_k^2 \frac{1}{k!}\Gamma(k-2\alpha+1)\big)}{\big(\sum_{k=k_0}^N w_k\frac{1}{k!} \Gamma(k-2\alpha)\big)^2}. &&(3) \tag{C.2.5}
\end{align*}

\medskip
\hrule
\end{minipage}
\caption{Formulas used in section C.2}
\end{table*}

Where $D_1=D_p(i,R)$, $D_m|x_{m-1}=D_p(x_{m-1},R)$ for $1<m<l-1$, and $D_{l-1}|x_{l-2}=D_p(x_{l-2},R)\cap D_p(j,R)$.
Note that all the integral areas, except the last one, are symmetric around node $i$ in angle. Also note that the last integral can be approximated as a linear function with respect to $\Delta \theta_{i,j}$ when $k_0$ is big. Consequently, the entire integral can be approximated as a linear function of $\Delta\theta_{i,j}$ as shown in equation C.2.2 in table 2.
Where $\xi_{m}$'s are imaginary nodes independent with $\Delta\theta_{i,j}$.\\
Looking at the variance, We first consider the bound for the number of weighted 4-hop paths.
Recall that the unweighted number of 4-hop paths has form of $S_4(i,j)=\sum_{x_1 \in \textrm{N}(i)}\sum_{x_2\in \textrm{N}(x_1)}\textrm{CN}(x_2,j)$, then the variance of the weighted version can be considered as:
{\fontsize{8.5}{8.5}\selectfont
\begin{align*}
\textrm{Var}\big( S_4^*(i,j) \big) &=\textrm{Var} \big( \sum_l w_l \sum_{\substack{x_1\in \textrm{N}(i)\\deg(x_1)=l}} \sum_k w_k \sum_{\substack{x_2\in \textrm{N}(x_1)\\deg(x_2)=k}}\textrm{CN}^*(x_2,j)\big)\\
                &\leq \textrm{Var}\big(\sum_l w_l \sum_{\substack{x_1\in \textrm{N}(i)\\deg(x_1)=l}} \sum_k w_k \sum_{\substack{x_2\in \textrm{N}(x_1)\\deg(x_2)=k}}\textrm{N}^*(j)\big)\\
                &= \textrm{Var}\big(\sum_l w_l \sum_{\substack{x_1\in \textrm{N}(i)\\deg(x_1)=l}}\sum_k w_k n_{1,k}\textrm{N}^*(j)\big) \quad \text{(1)}\\
                &= \textrm{Var}\big(\textrm{N}^*(i) \textrm{N}^*(\xi)\textrm{N}^*(j)\big), \quad \quad \quad \quad \quad \quad \quad \text{(2)}
\end{align*}}
In step (1), the $n_{1,k}$ indicate the number of nodes in $N(x_1)$ with degree $k$. In step (2), the $\xi$ is a imaginary point in $D_p(i,R)$.\\
By induction the variance of $S_l^*(i,j)$ can be bounded by:
$\textrm{Var}S_l^*(i,j)\leq \textrm{Var}(\textrm{N}^*(i)\textrm{N}^*(\xi_1)...\textrm{N}^*(\xi_{l-3})\textrm{N}^*(j)$ where $\xi_n\in D_p(\xi_{n-1},R)$ for $n>1$ and $\xi_1 \in D_p(i,R)$.\\
Note that the product of variance $\textrm{Var}\big(\textrm{N}^*(\xi_0)\textrm{N}^*(\xi_1)...\textrm{N}^*(\xi_{l-2})\big)=\prod_n \Big(\textrm{Var}\big(\textrm{N}^*(\xi_n)\big)+\mathbb{E}[\textrm{N}^*(\xi_n)]^2\Big)-\prod\Big(\mathbb{E}[\textrm{N}^*(\xi_n)]^2\Big)$.
The leading terms of this expression have form of $\Big(\prod_{n=0}^{l-3}\mathbb{E} [\textrm{N}^*(\xi_n)]^2 \Big)\cdot \textrm{Var}\big(\textrm{N}^*(\xi_{l-2})\big)$.\\
By rearranging the linear equation (C.2.2), the unbiased estimator for $\Delta \theta_{i,j}$ has form of equation C.2.3 in table 2.

Subsequently, we have the upper bound of the variance of the estimator as shown in C.2.4 in table 2.\\

To calculate the $\mathbb{E} [\textrm{N}^*(\xi)]$ and $\textrm{Var}\big(\textrm{N}^*(\xi)\big)$, define $n_{x,k}$ as the number of nodes in $D_p(x,R)$ with degree $k$, then:
{\fontsize{8.5}{8.5}\selectfont
\begin{equation*}
\begin{split}
\textrm{Var}(n_{x,k})&= \mathbb{E}(n_{x,k})\\
&=\frac{N}{2\pi}\int_{0}^{R}\, \rho(r)\, g(k|r)\, \theta_{x}(r)  dr\\
&=\frac{N}{2\pi}\int_{r_{x_1}}^{\hat{r}_{x_0}} \alpha e^{\alpha(r-R)}\frac{e^{-\bar{k}(r)}\bar{k}(r)^k}{k!}\Big(2e^{(R-r_x-r)/2}\Big)dr \\
&\approx \Theta(1)\frac{1}{k!}e^{(R-r_x)/2}\Gamma(k-2\alpha+1).
\end{split}
\end{equation*}}
Thus, we have $\mathbb{E} [\textrm{N}^*(\xi)]=\Theta(1)\sum_k w_k\frac{1}{k!}e^{(R-r_x)/2}\Gamma(k-2\alpha+1)$, and\\
$\textrm{Var}\big(\textrm{N}^*(\xi)\big)=\Theta(1)\sum_k w_k^2\frac{1}{k!}e^{(R-r_x)/2}\Gamma(k-2\alpha+1)$ as we assume each $n_{x,k}$ nearly independent.
Finally, the bound of variance can be derived in C.2.5 in table 2. Note that from step (1) to (2), the denominator terms $e^{(R-r_{\xi_m})/2}$ in (1) are reduced to $e^{(R-R)/2}$ and the numerator terms $e^{(R-r_{\xi_n})/2}$ are  magnified to $e^{R/2}$ since $r_{\xi}\in [0,R]$ for any ${\xi}$. This allows the first term in the numerator and denominator to cancel out to a constant in (3).\\
Following this simplification, we are able to minimize the variance bound and find that the optimal weight minimizing this variance bound should also be proportional to $\frac{1}{k}$.

\section{Supplemental Materials for Section 4.2}
\subsection{Calculation of the Integral in Section 4.2}
\begin{table*}
\centering
\begin{minipage}{0.75\textwidth}
\begin{align*}
    \mathbb{E}[S] 
    & \approx \frac{N}{2\pi}\sum_k w_k \int_{r_{x_1}}^{\hat{r}_{x_0}}\, \rho(r)\, g(k|r)\, \theta_{i,j}(r)  dr  \\
  &=\frac{N}{2\pi}\sum_k w_k \int_{r_{x_1}}^{\hat{r}_{x_0}} \alpha e^{\alpha(r-R)}\frac{e^{-\bar{k}(r)}\bar{k}(r)^k}{k!}\Big(2e^{(R-r_j-r)/2}+2e^{(R-r_i-r)/2}-\Delta \theta_{i,j}\Big)dr \\
  &=\frac{N}{2\pi}\sum_k w_k \int_{\bar{k}(\hat{r}_{x_0})}^{\bar{k}(\hat{r}_{x_1})} \Big[f_1(r_i,r_j)\bar{k}(r)^{-2\alpha}\frac{e^{-\bar{k}(r)}\bar{k}(r)^k}{k!}-\Theta(1)\bar{k}(r)^{-(2\alpha+1)}\frac{e^{-\bar{k}(r)}\bar{k}(r)^k}{k!}\Delta \theta_{i,j}\Big] d\bar{k}(r) &&\text{(1)}\\
  &\approx\sum_k w_k \frac{1}{k!}\Big[\Theta(1)f_1(r_i,r_j)\Gamma\big(k-2\alpha+1,\bar{k}(\hat{r}_{x_0})\big)-\Theta(1)\Gamma\big(k-2\alpha,\bar{k}(\hat{r}_{x_0})\big)\Delta\theta_{i,j}\Big], &&\text{(2)}
\end{align*}
\medskip
\hrule
\end{minipage}
\caption{The Equation of D.1.1}
\end{table*}
The calculation is shown as equation D.1.1 in table 3, where $\Gamma$ is the upper incomplete-gamma function, and $f_1(r_i,r_j)$ indicates a function of $r_i$ and $r_j$. In (1) we change the variable of the integral using equation(2) in section 3.3:
\begin{equation*}
\bar{k}(r)=\frac{4N\alpha}{\pi(2\alpha-1)}e^{-r/2}.
\end{equation*}
And in (2) we ignore the non-leading terms after integration.
\subsection{The Case Where $r_0<R$ in Section 4.2}
\textbf{Assume} $r_{x_0}<R$: \\
In this case  the association between the angle $\Delta \theta_{i,j}$ and $\mathbb{E}(n_k)$ is no longer linear, as the term $\Gamma\big(k-2\alpha,\bar{k}(r_{x_0})\big)$ contains $r_{x_0}$, which depends on $\Delta \theta_{i,j}$.
Consequently, we are unable to derive an explicit expression for $\Delta \theta_{i,j}$ in terms of $n_k$. For sufficiently large values of $k$, however, the incomplete upper gamma function $\Gamma(k-2\alpha,\bar{k}(r_{x_0}))$ becomes close to the complete gamma function $\Gamma(k-2\alpha)$. This approximation results in a linear relationship between $\Delta \theta_{i,j}$ and $\mathbb{E}(n_k)$ once again. Therefore, we can use the same reasoning as if $r_{x_0}\geq R$. This is equivalent to finding an optimal weight for an estimator $S_{k_0}=\sum_{k=k_0}^N w_k n_k$ for sufficiently large $k_0$. In this case, we can obtain the same weight $w_{k^*}\sim\frac{1}{k^*}$.
\subsection{Optimizing the Variance in Section 4.2}
The variance of $\hat{f}(S)$ is:
\begin{equation*}
\begin{split}
\textrm{Var}(\hat{f}(S)) &= \frac{\sum w_k^2 \textrm{Var}(n_k)}{[\Theta(1)\sum_k w_k\frac{1}{k!} \Gamma\big(k-2\alpha,\bar{k}(R)\big)]^2} ,
\end{split}
\end{equation*}
To minimize the variance of the unbiased estimator, we set $\frac{d\, \textrm{Var}(\hat{f}(S))}{d w_{k^*}}=0$, and get:
\begin{equation*}
\begin{split}
&\Theta(1)[\sum_{k} w_{k} \frac{1}{k!} \Gamma(k-2\alpha,\bar{k}(R))]^2\cdot 2w_{k^*}\textrm{Var}(n_{k^*})\\
&-\sum_{k_1} w_{k_1}^2 \textrm{Var}(n_{k_1})\cdot2[\Theta(1)\sum_{k_2} w_{k_2}\frac{1}{k_2!} \Gamma(k_2-2\alpha,\bar{k}(R)]\\
&\cdot \frac{1}{k^*!}\Gamma({k^*}-2\alpha,\bar{k}(R))=0.
\end{split}
\end{equation*}
By solving this equation we get:
$w_{k^*} \propto \frac{\Gamma({k^*}-2\alpha,\bar{k}(R))}{k^*!\textrm{Var}(n_{k^*})}= \frac{\Gamma({k^*}-2\alpha,\bar{k}(R))}{\Gamma(k^*-2\alpha+1,\bar{k}(R))-\Theta(1)\Gamma(k^*-2\alpha,\bar{k}(R))\Delta \theta}\sim \frac{1}{k^*}$ as $k^*\to \infty$.\\
\section{The Katz Index}
Closely related to the indices depending on the $l$-hops paths, the Katz index depends on the set of all paths between nodes $i$ and $j$. It can be thought of as a linear combination of the $l$-hop paths indices:
\begin{equation}
\label{eq:katz}
\textrm{Katz}(i,j)=\sum_{l=1}^\infty \beta^l|\textrm{paths}^{l}_{i,j}|, \tag{E.1}
\end{equation}
where $\beta$ is a damping factor.
Naturally, we can also define a weighted version of this:\\
\textbf{The weighted Katz index}\\
Consider a weighted variant of the Katz index:
\begin{equation}
\label{eq:wkatz}
\textrm{Katz}^*(i,j)=\sum_{l=1}^\infty \beta^l(\sum \textrm{path}^{*l}_{i,j}),
\tag{E.2}
\end{equation}
We choose the degree weight to be $\frac{1}{k}$ as before. In practice, the index could be expensive to calculate for a large network since it involves the calculation of a matrix inverse. A common approach is to localize the index \cite{zhou2009predicting,lu2009similarity}:\\
\textbf{The localized weighted Katz index}
\begin{equation}
\label{eq:lwkatz}
\textrm{Katz}^*(i,j)\textrm{-}t=\sum_{l=1}^t \beta^l(\sum \textrm{path}^{*l}_{i,j}).
\tag{E.3}
\end{equation}

\section{Correlations and Link Prediction experiments on HRGs}
\subsection{Correlation Test}
For synthetic graphs, we supplement the main paper's experiments by including 6 additional HRGs (indices 7-12; 1-6 are used in main paper). The parameters used are shown in table 4. For each index, we calculate Spearman's correlation coefficient (detailed in Appendix H) relative to the true latent distance as an evaluation metric. \\
The following indices are included in the experiments:
\textbf{CN}: Number of common neighbors;
\textbf{AA}: The Adamic/Adar index;
\textbf{RA}: The Resource Allocation index;
\textbf{GD}: The Graph distance;
\textbf{HGD}: The Hyperbolic Graph distance ;
\textbf{$l$-hop}: Number of $l$-hop paths, where $l\textrm{-hop}^*$ means it is weighted by $\frac{1}{k}$;
\textbf{Katz}: The Katz index (\ref{eq:katz}), where $\textrm{Katz}^*$ means it is weighted by $\frac{1}{k}$ (\ref{eq:wkatz});.
\textbf{Katz-t}: The Localized Katz index, where $\textrm{Katz*-t}$ means it is weighted by $\frac{1}{k}$ and contains at most $t$-hop paths (\ref{eq:lwkatz}). For the Katz index, we let the damping factor, $\beta$, range from 0.001 to 0.5, and then record the best result. For each experiment, we independently run 10 times and report the average score.

\begin{table*}[h]
\vspace{2mm}
\centering
\begin{tabular}{lllllllllllll} 
\hline
  & 1   & 2   & 3    & 4    & 5    & 6    & 7   & 8   & 9    & 10   & 11   & 12    \\ 
\hline
N & 500 & 500 & 1000 & 1000 & 2000 & 2000 & 500 & 500 & 1000 & 1000 & 2000 & 2000  \\
R & 8   & 10  & 10   & 12   & 10   & 12   & 10  & 8   & 12   & 10   & 12   & 10    \\
$\alpha$ & 0.8 & 1   & 0.8  & 1    & 1    & 1.2  & 0.8 & 1   & 0.8  & 1    & 1    & 1.2   \\
\hline
\end{tabular}
\caption{Synthetic networks generated by HRG with different parameters.}
\end{table*}
\begin{table*}[h]
\vspace{2mm}
\centering
\resizebox{1\textwidth}{!}{
\begin{tabular}{lllllllllllll} 
\hline
       & 1                & 2                & 3                & 4                & 5                & 6                & 7                & 8                & 9                & 10               & 11               & 12                \\ 
\hline
CN     & -0.8996          & -0.5340          & -0.8043          & -0.4279          & -0.7079          & -0.3329          & -0.6933          & -0.7930          & -0.5553          & -0.5922          & -0.4346          & -0.5463           \\
AA     & -0.9070          & -0.5358          & -0.8172          & -0.4286          & -0.7124          & -0.3331          & -0.7036          & -0.8060          & -0.5587          & -0.5977          & -0.4357          & -0.5473           \\
RA     & \textit{-0.9098} & \textit{-0.5360} & \textit{-0.8217} & \textit{-0.4286} & \textit{-0.7137} & \textit{-0.3331} & \textit{-0.7057} & \textit{-0.8121} & \textit{-0.5588} & \textit{-0.5985} & \textit{-0.4357} & \textit{-0.5474}  \\ 
\hline
GD     & 0.6990           & 0.5865           & 0.7010           & 0.5741           & 0.6957           & 0.4608           & 0.6089           & 0.6588           & 0.5670           & 0.5919           & 0.5136           & 0.6136            \\
HGD    & \textbf{0.9704}  & \textbf{0.8396}  & \textbf{0.9533}  & \textbf{0.7959}  & \textbf{0.9553}  & \textbf{0.8102}  & \textbf{0.9001}  & \textbf{0.9517}  & \textbf{0.8717}  & \textbf{0.9240}  & \textbf{0.8703}  & \textbf{0.9426}   \\ 
\hline
3-hop   & -0.9058          & -0.6054          & -0.8644          & -0.5405          & -0.8267          & -0.4814          & -0.7425          & -0.8562          & -0.6612          & -0.6729          & -0.5526          & -0.6919           \\
3-hop*  & -0.9580          & -0.6481          & -0.9263          & -0.5614          & -0.8736          & -0.4891          & -0.8446          & -0.9178          & -0.7321          & -0.7435          & -0.5881          & -0.7291           \\
4-hop   & -0.8403          & -0.6046          & -0.8118          & -0.5668          & -0.8110          & -0.5024          & -0.6827          & -0.7989          & -0.6274          & -0.6552          & -0.5567          & -0.7006           \\
4-hop*  & \textit{-0.9618} & \textit{-0.6899} & -\textit{0.9383} & -\textit{0.6093} & -\textit{0.8946} & -\textit{0.5266} & -\textit{0.8686} & -\textit{0.9308} & -\textit{0.7837} & \textit{-0.7788} & \textit{-0.6306} & \textit{-0.7655}  \\ 
\hline
katz-3  & -0.9252          & -0.6094          & -0.8723          & -0.5421          & -0.8267          & -0.4815          & -0.7558          & -0.8581          & -0.6632          & -0.6745          & -0.5530          & -0.6919           \\
katz*-3 & -0.9492          & -0.6455          & -0.9098          & -0.5619          & -0.8678          & -0.4890          & -0.8295          & -0.8993          & -0.7257          & -0.7369          & -0.5875          & -0.7286           \\
katz-4  & -0.9247          & -0.6200          & -0.8675          & -0.5718          & -0.8249          & -0.5027          & -0.7522          & -0.8539          & -0.6590          & -0.6769          & -0.5639          & -0.7029           \\
katz*-4 & -0.9540          & -0.6785          & -0.9194          & -0.6058          & -0.8805          & -0.5231          & -0.8452          & -0.9087          & -0.7571          & -0.7581          & -0.6217          & -0.7575           \\
katz   & -0.9246          & -0.6137          & -0.8672          & -0.5806          & -0.8241          & -0.4776          & -0.7520          & -0.8538          & -0.6578          & -0.6763          & -0.5581          & -0.7017           \\
katz*  & \textit{-0.9586} & \textit{-0.6976} & \textit{-0.9288} & \textit{-0.6442} & \textit{-0.8922} & \textit{-0.5362} & \textit{-0.8584} & \textit{-0.9181} & \textit{-0.7791} & \textit{-0.7757} & \textit{-0.6420} & -\textit{0.7716}  \\
\hline
\end{tabular}
}
\caption{Spearman's correlation coefficient of indices in synthetic networks}
\end{table*}
Table 5 shows results consistent with previous observations: the weighted version outperforms the unweighted version in HRG, and HGD excels in estimating latent distances. This reinforces the effectiveness of the weighted approach and HGD's superior performance in these settings, and verifies our theoretical analysis. 
\subsection{Link Prediction}
The evaluation of link prediction performance follows a widely adopted approach. Given a graph $G = (V, E)$, we partition the edge set $E$ into two sets: the training set $E_t$, and the probe set $E_p$, which consists of the missing links. Only the information contained in $E_t$ is utilized to predict the missing links in $E_p$. In our experiment, we remove $10\%$ of the total edges. To measure how well these indices can predict the missing links, we choose average precision among numerous metrics (Appendix H explains why choosing average precision over other metrics). The results in Table 6 are consistent with those presented in the main paper.

\begin{table*}[h]
\vspace{2mm}
\centering
\resizebox{1\textwidth}{!}{
\begin{tabular}{lllllllllllll} 
\hline
       & 1               & 2               & 3               & 4               & 5               & 6               & 7               & 8               & 9               & 10              & 11              & 12               \\ 
\hline
CN     & 0.5430          & 0.3413          & 0.4950          & 0.2907          & 0.6286          & 0.4077          & 0.3578          & 0.5731          & 0.2914          & 0.4905          & 0.4149          & 0.6337           \\
AA     & 0.6897          & 0.5273          & 0.6679          & 0.4818          & 0.7461          & 0.5774          & 0.5398          & 0.7003          & 0.5064          & 0.6471          & 0.5914          & 0.7442           \\
RA     & 0.8041          & 0.5879          & 0.7596          & 0.5108          & \textbf{0.8141} & \textbf{0.6331} & 0.6303          & \textbf{0.7662} & 0.5611          & \textbf{0.7286} & \textbf{0.6528} & \textbf{0.8091}  \\ 
\hline
GD     & 0.0092          & 0.0162          & 0.0058          & 0.0189          & 0.0056          & 0.0210          & 0.0074          & 0.0133          & 0.0061          & 0.0123          & 0.0135          & 0.0067           \\
HGD    & 0.6463          & 0.3652          & 0.5346          & 0.3053          & 0.4435          & 0.2914          & 0.4570          & 0.4659          & 0.3930          & 0.3846          & 0.3411          & 0.3627           \\ 
\hline
3-hop   & 0.3334          & 0.1779          & 0.2514          & 0.1442          & 0.2940          & 0.1823          & 0.1901          & 0.3128          & 0.1555          & 0.2479          & 0.1885          & 0.3056           \\
3-hop*  & 0.7186          & 0.4739          & 0.6922          & 0.3923          & 0.7312          & 0.5046          & 0.5441          & 0.6775          & 0.5072          & 0.6302          & 0.5494          & 0.7043           \\
4-hop   & 0.2507          & 0.1622          & 0.1681          & 0.1049          & 0.1643          & 0.1140          & 0.1410          & 0.2075          & 0.1188          & 0.1582          & 0.1113          & 0.1646           \\
4-hop*  & 0.6122          & 0.4869          & 0.6178          & 0.4215          & 0.6525          & 0.5034          & 0.5209          & 0.6138          & 0.4967          & 0.5887          & 0.5244          & 0.6279           \\ 
\hline
katz-3  & 0.5498          & 0.3508          & 0.5038          & 0.3031          & 0.6364          & 0.4110          & 0.3716          & 0.5842          & 0.3081          & 0.5007          & 0.4244          & 0.6404           \\
katz*-3 & \textbf{0.8068} & \textbf{0.5902} & \textbf{0.7632} & \textbf{0.5143} & \textbf{0.8141} & \textbf{0.6334} & \textbf{0.6334} & \textbf{0.7665} & \textbf{0.5732} & \textbf{0.7290} & \textbf{0.6532} & \textbf{0.8090}  \\
katz-4  & 0.5492          & 0.3511          & 0.5036          & 0.3036          & 0.6361          & 0.4111          & 0.3716          & 0.5841          & 0.3080          & 0.5007          & 0.4245          & 0.6404           \\
katz*-4 & \textbf{0.8065} & 0.5896          & 0.7624          & \textbf{0.5139} & \textbf{0.8141} & \textbf{0.6334} & 0.6326          & \textbf{0.7665} & 0.5714          & \textbf{0.7290} & \textbf{0.6531} & \textbf{0.8090}  \\
katz   & 0.5491          & 0.3511          & 0.5036          & 0.3036          & 0.6361          & 0.4111          & 0.3716          & 0.5841          & 0.3080          & 0.5007          & 0.4245          & 0.6404           \\
katz*  & \textbf{0.8065} & 0.5896          & 0.7619          & \textbf{0.5139} & \textbf{0.8141} & \textbf{0.6334} & 0.6322          & \textbf{0.7665} & 0.5709          & \textbf{0.7290} & \textbf{0.6531} & \textbf{0.8090}  \\
\hline
\end{tabular}
}
\caption{AP (average precision) score of indices in synthetic networks.}
\end{table*}

\section{Link Prediction Experiments on Real World Graphs}
In this section, we conduct extensive experiments on real networks to evaluate the performance of the link prediction indices. In selecting network datasets for these experiments, we ensured to cover a wide range of fields:\\
\textbf{Karate}: A friendship network comprising 34 members of a Karate club; \textbf{Dolphin}: A social network representing the interactions among 62 bottlenose dolphins; \textbf{USAir97}: An airline network depicting connections between US airports in 1997; \textbf{Netscience}: A coauthorship network of 379 researchers in the field of network theory; \textbf{C.elegans}: A network involving biochemical reactions occurring within the cells of the model organism Caenorhabditis elegans (C.elegans); \textbf{Harvard}: A social friendship network extracted from Facebook, consisting of 1446 people affiliated with Harvard, with edges indicating friendship ties. \textbf{NS}: A collaboration network of researchers in network science with 1,589 nodes and 2,742 edges which is a bigger dataset than Netscience dataset. \textbf{PB}: A network of US political blogs with 1222 nodes and 16714 edges. \textbf{Yeast}: A protein-protein interaction network in yeast with 1458 nodes and 1948 edges. \textbf{Cora}: A dataset consists of 2708 scientific publications classified into one of seven classes. The citation network consists of 5278 links. \textbf{Citeseer}: A dataset contains 3264 scientific publications grouped into one of six categories consisting of 4536 links.

\begin{table}[h]
\vspace{2mm}
\centering
\begin{tabular}{llllll} 
\hline
           & $|V|$    & $E$     & $\bar{k}$           & C      & D       \\ 
\hline
Karate     & 34   & 78    & 4.588       & 0.571 & 2.408   \\
Dolphin    & 62   & 159   & 5.129       & 0.259  & 3.357   \\
USair97    & 332  & 2126  & 12.807      & 0.625  & 2.738   \\
Netscience & 379  & 914   & 4.823       & 0.741  & 6.042   \\
C.elegans   & 454  & 2040  & 8.986       & 0.645  & 2.660    \\
Harvard    & 1446 & 59589 & 82.419      & 0.323  & 2.228   \\
PB         & 1222 & 16717 & 27.360 & 0.3202 & 2.737   \\
NS         & 1589 & 2742  & 3.451      & 0.6377 & 6.042  \\
Yeast      & 1458 & 1948  & 2.672      & 0.0708 & 6.812  \\
Cora       & 2708 & 5278  & 3.898       & 0.240 & 6.311  \\
Citeseer   & 3264 & 4536  & 2.779      & 0.144 & 9.310  \\
\hline
\end{tabular}
\caption{Properties of the real networks. $|V|$ and $|E|$ are the total numbers of nodes and edges of the networks. $\bar{k}$ represents the average degree, $C$, the clustering coefficient, and $D$, the average path length.}
\end{table}
\begin{table*}[h]
\vspace{2mm}
\centering
\resizebox{1\textwidth}{!}{
\begin{tabular}{llllllllllll} 
\hline
        & Karate          & Dolphin         & USair97         & Netscience      & C.elegans        & Harvard         & PB              & NS              & Yeast           & Cora            & Citeseer         \\ 
\hline
CN      & 0.1122          & 0.0949          & 0.3363          & 0.3218          & 0.0663          & 0.1111          & 0.0931          & 0.4257          & 0.0039          & 0.0131          & 0.0086           \\
AA      & 0.1349          & 0.1170          & 0.3728          & 0.5641          & 0.1312          & 0.1227          & 0.0920          & 0.6761          & 0.0059          & 0.0229          & 0.0137           \\
RA      & 0.1317          & 0.1011          & 0.4244          & 0.5905          & 0.1756          & \textbf{0.1552} & 0.0754          & 0.6978          & 0.0055          & 0.0216          & 0.0128           \\ 
\hline
GD      & 0.0209          & 0.0238          & 0.0109          & 0.0310          & 0.0052          & 0.0087          & 0.0060          & 0.0591          & 0.0008          & 0.0037          & 0.0046           \\
HGD     & 0.1269          & 0.0355          & 0.3225          & 0.0865          & 0.0655          & 0.0267          & 0.0325          & 0.1001          & 0.0011          & 0.0028          & 0.0045           \\ 
\hline
3-hop   & 0.0883          & 0.0397          & 0.3287          & 0.1279          & 0.0821          & 0.0559          & \textbf{0.1082} & 0.2611          & 0.0059          & 0.0082          & 0.0162           \\
3-hop*  & \textbf{0.1643} & 0.0582          & 0.4164          & 0.3624          & 0.1243          & 0.0837          & 0.1020          & 0.4990          & \textbf{0.0105} & 0.0248          & \textbf{0.0320}  \\
4-hop   & 0.1129          & 0.0697          & 0.3071          & 0.1054          & 0.0577          & 0.0442          & 0.0944          & 0.2164          & 0.0021          & 0.0074          & 0.0062           \\
4-hop*  & 0.1444          & 0.1126          & 0.3836          & 0.5135          & 0.1200          & 0.0617          & 0.0743          & 0.6444          & 0.0050          & 0.0243          & 0.0145           \\ 
\hline
katz-3  & 0.1583          & 0.1143          & 0.3476          & 0.3288          & 0.0869          & 0.1069          & \textbf{0.1083} & 0.4293          & 0.0070          & 0.0185          & 0.0177           \\
katz*-3 & 0.1523          & 0.1092          & \textbf{0.4460} & 0.5980          & \textbf{0.1774} & \textbf{0.1551} & 0.0844          & 0.7049          & 0.0076          & 0.0273          & 0.0185           \\
katz-4  & 0.1513          & \textbf{0.1218} & 0.3470          & 0.3287          & 0.0786          & 0.1060          & 0.1007          & 0.4276          & 0.0070          & 0.0184          & 0.0128           \\
katz*-4 & 0.1541          & 0.1105          & 0.4412          & \textbf{0.5987} & 0.1767          & \textbf{0.1551} & 0.0837          & 0.7066          & 0.0072          & 0.0276          & 0.0185           \\
katz    & 0.1513          & \textbf{0.1217} & 0.3470          & 0.3287          & 0.0782          & 0.1058          & 0.0993          & 0.4275          & 0.0070          & 0.0185          & 0.0139           \\
katz*   & 0.1514          & 0.1143          & 0.4363          & \textbf{0.5987} & 0.1766          & \textbf{0.1551} & 0.0836          & \textbf{0.7080} & 0.0074          & \textbf{0.0285} & 0.0194           \\
\hline
\end{tabular}
}
\caption{AP (average precision) score of indices in real networks}
\end{table*}
The basic statistics of these networks are provided in Table 7. In Table 8, additional experiments show that the weighted schemes generally outperform their unweighted counterparts, with a few exceptions. Notably, HGD consistently outperforms standard graph distance in most cases, except for the Cora dataset and the Citeseer dataset, where performance is nearly identical. These results support the effectiveness of our analysis and proposed modifications on real-world networks. Furthermore, we observe that incorporating longer-hop paths does not significantly enhance link prediction performance.
\section{Evaluation Metrics: the Spearman correlation coefficient and Average precision}
\subsection{Spearman Correlation Coefficient}
As previously discussed, the link prediction task does not require an accurate estimation of the latent distance. Instead, a weaker ordinal relationship is sufficient for effective link prediction. To assess whether an index reflects the ordinal relation between distances, we can utilize the Spearman correlation coefficient.
The Spearman correlation coefficient measures how well the relationship between two variables can be described using a monotonic function. It ranges from -1 to 1, with higher values indicating a stronger correlation. Intuitively, a high Spearman correlation suggests that the index and the true distance have a similar rank (a perfect rank match has correlation of 1), indicating that the index is well-suited for the link prediction task.\\
\subsection{Average Precision}
ROC curve (receiver operating characteristic curve) and PR curve (Precision-recall curve ) are two widely common used metric for classification task. The ROC curve is typically used in balanced classification tasks, where the class distribution is relatively equal. However, in link prediction, the number of positive samples (missing edges) are far less than the number of negative samples(non-connected edges), resulting in a heavily imbalanced dataset. In contrast, the PR curve, which focuses on the trade-off between precision (the fraction of correctly predicted positive examples) and recall (the fraction of actual positive examples predicted correctly), provides a more accurate representation of the model's performance in this skewed scenario.
precision-recall curve is a graph showing the performance of a classification model at all classification thresholds. This curve plots two parameters:\\
Precision
$$
\textrm{PR}=\frac{\textrm{TP}}{\textrm{TP}+\textrm{FP}}.
$$
Recall
$$
\textrm{RC}=\frac{\textrm{TP}}{\textrm{TP}+\textrm{FN}}.
$$
AUPR is a single number used to quantify the quality of  a Precision-Recall curve,
The area under precision-recall curve (AUPR) is then:
$$
\textrm{AUPR}=\int_0^1 \textrm{PR}(t)\textrm{RC}(t)dt.
$$
There are multiply ways to approximate AUPR. J.Davis and M.Goadrich \citep{davis2006relationship} pointed out that using trapezoid rule to approximate AUPR may cause overly optimistic result, a more reasonable approximation is to calculate average precision :
$$
\textrm{AP}=\sum (\textrm{RC}_{n}-\textrm{RC}_{n-1})\textrm{PR}_n.
$$
where $\textrm{PR}_n$ and $\textrm{RC}_{n}$ are the precision and recall at the $n$th threshold. The average precision can be think as a weighted mean of the precisions achieved at each threshold. AP ranges between 0 and 1, with a higher value indicating better performance of the predictor.

\section{Graph embedding with h-MDS}
In this section, we take h-MDS, a hyperbolic variant of multidimensional scaling, as an example, replacing its graph distance with HGD to demonstrate improved mean average precision (MAP) on tree-structured datasets.\\
\subsection{Background}
The hyperboloid model, commonly known as the Minkowski model, is isometric to the Poincaré model. In this model, $n$-dimensional points are mapped onto the two-sheeted hyperboloid in $(n+1)$-dimensional Minkowski space. This model embeds hyperbolic space isometrically into Minkowski space, allowing the hyperbolic distance to be derived from the underlying Minkowski space. This is analogous to deriving spherical distances from Euclidean distances in the case of an $n$-sphere embedded in $(n+1)$-dimensional Euclidean space.\\
Let $x_1=(x_{0,1},x_{1,1},...,x_{n,1})$ and $x_2=(x_{0,2},x_{1,2},...,x_{n,2})$ be two vectors in $\mathbb{R}^{n+1}$, then the Minkowski bilinear form is defined as:
$$
<x_1,x_2>_L=\sum_{i=1}^{n} x_{i,1}x_{i,2}-x_{0,1}x_{0,2}.
$$
Sometimes the above bilinear form is called Lorenz inner product. Note that unlike Euclidean dot product, the bilinear form is not positive definite, i.e. there exist a vector $z$, such that $<z,z>_l<0$. Actually, the $n$-dimensional hyperboloid $\mathbb{H}^n$ is defined as a collection of such points:
$$
\mathbb{H}^n=\{x\in \mathbb{R}^{n+1}|<x,x>_L=-1, x_0>0\}.
$$
The distance between two points on the hyperboloid can be calculated by:
$$
d_{\mathbb{H}^n}(x_1,x_2)=\text{ar\text{cosh}} (-<x_1,x_2>_L),
$$
for any $x_1,x_2\in \mathbb{H}^n$.\\
Suppose there are points $x_1,...,x_n \in \mathbb{H}^{n}$. We can compute the scaled distance matrix $Y$ such as:
\[
Y_{i,j}=\textrm{cosh}(d_{i,j})=x_{0,i}x_{0,j}-\hat{x}_i^T\hat{x}_j,
\]
where $\hat{x}_i=(x_{1,i},...,x_{n,i})$.
Define $X=(\hat{x}_1,...,\hat{x}_n)$ and $u=(x_{0,1},...x_{0,n})$. Then we can write Y in matrix form as
\[
Y=uu^T-XX^T.
\]
Without loss of generality, we can suppose that $X^Tu=0$, which is defined as pseudo-Euclidean mean in \citep{sala2018representation}.\\
This implies that $u$ is an eigenvector of $Y$ with positive eigenvalue, and the rest of $Y$ ’s eigenvalues are negative whose corresponding eigenvectors is $X$. Therefore by eigendecomposition of $Y$ we can recover $X$ up to rotation; See algorithm \ref{alg:hmds} for details. Note that the original algorithm uses standard graph distance to calculate the distance matrix.\\
\begin{algorithm}
\caption{h-MDS}\label{alg:hmds}
\begin{algorithmic}[1]
\State \textbf{Input: Distance matrix $d_{ij}$ and rank $r$}
\State Compute scaled distance matrix $Y=\textrm{cosh}(d_{i,j})$
\State Eigenvalue decomposition ($Y,r$) $\rightarrow X$
\State Project $X$ from hyperboloid model to Poincaré model: $x \rightarrow \frac{x}{1+\sqrt{1+||x||^2}}$
\State If desired, center $X$ at a different mean
\State \textbf{Return $X$}
\end{algorithmic}
\end{algorithm}
\subsection{Metric and Experiments}
We use a widely used metric, mean average precision (MAP) to measure the performance of embedding, in terms of preserving local similarity. For a graph $G=(V,E)$, let $a\in V$ have neighborhood $N_a=\{b_1,b_2,...\}$. In the embedding $f$, consider the points cloest to $f(a)$, and define $R_{a,b_i}$ to be the smallest set of such points that contains $b_i$ ($R_{a,b_i}$ is the smallest set of nearest points required to include the $i$th neighbor of $a$ in $f$). Then the MAP is defined as:
\[
\textrm{MAP}(f)=\frac{1}{|V|}\sum_{a\in V} \frac{1}{\textrm{deg}(a)}\sum_{i=1}^{|N_a|}\frac{|N_a \cap R_{a,b_i}|}{|R_{a,b_i}|}.
\]
The tree-structured datasets used include: a small and fully-balanced tree consisting of 40 nodes; a phylogenetic tree consisting of 344 nodes\citep{sanderson1994treebase}\citep{hofbauer2016preliminary}; a CS-PhD graph of the relations between advisors and PhD students consisting of 1025 nodes\citep{nooy2011exploratory}; A transitive closure of the WordNet mammals subset, consisting of 2250 nodes.\\
Table 3 in the main paper shows that incorporating HGD into graph embedding methods can improve the performance of embedding representations in hierarchical datasets.